\documentclass[sigconf,nonacm]{acmart}

\usepackage{url}
\usepackage{nicefrac}
\usepackage{makecell}
\usepackage[ruled,vlined]{algorithm2e}
\usepackage{multirow}
\usepackage{amsfonts}
\usepackage{enumitem}
\usepackage{subcaption}
\usepackage{bm}

\newtheorem{theorem}{Theorem}[section]

\usepackage{colortbl}  
\newcommand{\timerow}{\rowcolor{gray!12}\selectfont}  
\newcommand{\spara}[1]{\smallskip\noindent{\bf #1}}

\newenvironment{squishlist}
{\begin{list}{$\bullet$}
 {\setlength{\itemsep}{0pt}
  \setlength{\parsep}{1pt}
  \setlength{\topsep}{2pt}
  \setlength{\partopsep}{0pt}
  \setlength{\leftmargin}{1.5em}
  \setlength{\labelwidth}{1em}
  \setlength{\labelsep}{0.5em}}}
{\end{list}}

\AtBeginDocument{%
  }

\copyrightyear{2026}
\acmYear{2026}
\setcopyright{cc}
\setcctype{by}
\acmConference[CIKM '26]{Proceedings of the 35th ACM International Conference on Information and Knowledge Management}{November 07--11, 2026}{Rome, Italy}
\acmBooktitle{Proceedings of the 35th ACM International Conference on Information and Knowledge Management (CIKM '26), November 07--11, 2026, Rome, Italy}
\acmDOI{10.1145/3799682.3840960}
\acmISBN{979-8-4007-2539-5/2026/11}

\begin{document}

\title{JECHO: Scalable Echo Chamber Detection via Jaccard-based Homophily and Seed Expansion}
\thanks{Accepted to CIKM 2026.}
\author{Ali Safarpoor Dehkordi}
\orcid{0009-0006-3876-2926}
\affiliation{%
  \institution{The Australian National University}
  \city{Canberra}
  \state{Australian Capital Territory}
  \country{Australia}
}
\email{ali.safarpoordehkordi@anu.edu.au}

\author{Atsushi Miyauchi}
\orcid{0000-0002-6033-6433}
\affiliation{%
  \institution{Intesa Sanpaolo}
  \city{Turin}
  \country{Italy}
}
\email{atsushi.miyauchi@intesasanpaolo.com}

\author{Francesco Bonchi}
\orcid{0000-0001-9464-8315}
\affiliation{%
  \institution{Intesa Sanpaolo AI Research}
  \city{Turin}
  \country{Italy}
}
\email{francesco.bonchi@intesasanpaolo.com}

\author{Ahad N. Zehmakan}
\correspondingauthor
\orcid{0000-0002-8569-6347}
\affiliation{%
  \institution{The Australian National University}
  \city{Canberra}
  \state{Australian Capital Territory}
  \country{Australia}
}
\email{ahadn.zehmakan@anu.edu.au}

\renewcommand{\shortauthors}{Safarpoor Dehkordi et al.}

\begin{abstract}
Detecting echo chambers is critical for understanding and limiting negative social phenomena, such as online polarization, misinformation, and conspiracy theory diffusion. However, the echo chamber detection (ECD) problem yet lacks a rigorous algorithmic foundation. We address this gap by formalizing a unified definition of echo chambers based on three principles: internal opinion homogeneity, opinion extremism, and structural isolation. Under such a definition, we establish the theoretical hardness of the ECD problem, proving it is NP-hard via a reduction from the conductance minimization problem. To circumvent this computational barrier, we derive a theoretical lower bound on the Jaccard-based homophily (JHO) of nodes that reside within valid echo chambers. This guarantee motivates JECHO, a novel algorithm that detects echo chambers via local seed expansion rather than global enumeration. JECHO first identifies seeds that exceed the JHO threshold and then applies a score-based expansion to optimize structural isolation. 
Extensive experiments on real-world and synthetic networks demonstrate that our theory-guided approach detects more structurally isolated echo chambers than state-of-the-art methods while reducing runtime by orders of magnitude.
\end{abstract}




\maketitle \sloppy

\section{Introduction}
Online social networks play a central role in shaping political discourse, information exposure, and collective opinion formation. A key phenomenon on these platforms is the formation of echo chambers, that is, groups of users that are ideologically aligned and structurally isolated from the rest of the network. Echo chambers limit exposure to opposing views, amplify misinformation, and contribute to growing social and political polarization~\cite{del2016echo,choi2020rumor,jiang2021mechanisms}.

Empirical studies have demonstrated the persistence and impact of echo chambers across topics and times. For example, Nasim et al.~\cite{nasim2022we} showed that users involved in polarized debates in Australia repeatedly reappear in later discussions while maintaining similar ideological positions. 
Jiang et al.~\cite{jiang2021mechanisms} further identify key attributes of echo chambers, including the diffusion of misinformation, the spread of conspiracy theories, political polarization, and emotional contagion across various social media contexts.
Such findings have motivated a growing body of work on echo chamber detection (ECD), where the goal is to efficiently identify node sets that form echo chambers in a given input network with node-level opinions~\cite{villa2021echo,impicciche2024comparing,perera2024quantifying}.

However, existing approaches suffer from two fundamental limitations. 
Firstly, many existing formulations are incomplete and capture only limited characteristics of echo chambers, such as opinion homogeneity, opinion extremism, or structural isolation~\cite{mahmoudi2024echo,hartmann2025systematic}, limiting their effectiveness and applicability. 
Secondly, many ECD methods rely on global graph partitioning, community detection, or network-wide opinion modeling which limits their scalability to large social networks, as noted in recent surveys (cf.~\cite{mahmoudi2024echo}).

Echo chambers are inherently complex social phenomena that do not have a single, universally accepted mathematical formulation. Prior work has characterized them using different combinations of ideological alignment, selective exposure, and structural segregation~\cite{mahmoudi2024echo,hartmann2025systematic}. Rather than aiming to formulate this phenomenon in its full real-world complexity, which may not be feasible, our goal is to provide a principled and mathematically rigorous formulation grounded in three core characteristics consistently identified in the prior literature: \emph{internal opinion homogeneity}~\cite{del2016echo,cinelli2021echo}, \emph{opinion extremism}~\cite{bright2018explaining,sunstein2002law}, and \emph{structural isolation}~\cite{cinelli2021echo}.

Building on this formulation, we formalize the ECD problem and prove its NP-hardness, establishing a fundamental computational barrier for exact computation. Despite this hardness result, we show that echo chambers exhibit a key structural property: 
Under practically motivated conditions specified in Section~\ref{ssec:seed-theory}, we show that an echo chamber contains at least one node whose \emph{Jaccard-based homophily (JHO)} score exceeds a theoretically derived threshold.
This result provides a principled way to identify echo chambers through local searches from seed nodes rather than global enumeration. Based on this insight, we propose JECHO, a two-stage ECD framework.
JECHO first uses the JHO score to identify a small set of candidate seed nodes that are likely to belong to echo chambers. A score-based seed expansion (SSE) procedure is then applied to efficiently recover full echo chambers while preserving opinion homogeneity, extremism, and structural isolation. This allows us to detect the desired echo chambers efficiently and effectively, without scanning the entire graph. Extensive experiments demonstrate that JECHO achieves higher detection quality while reducing runtime by orders of magnitude compared to state-of-the-art baselines, such as SEDA~\cite{perera2024quantifying} and EchoGAE~\cite{alatawi2023quantifying}.

One advantage of our framework is its flexibility with respect to the specific metrics used to instantiate these characteristics. We empirically demonstrate that JECHO is largely insensitive to the choice of metric formulations and threshold parameters, with broadly stable overall trends observed across configurations.


Our main contributions are summarized as follows:
\begin{squishlist}

\item \textbf{Unified Definition:}
We introduce a unified definition of echo chambers that integrates opinion homogeneity, opinion extremism, and structural isolation, based on which we formulate a novel ECD problem (Section~\ref{s:prel}).

\item \textbf{Computational Complexity Analysis:}
We prove that the ECD problem is NP-hard via a reduction from the conductance minimization problem~\cite{raghavendra2012reductions} (Section~\ref{s:prel}).

\item \textbf{JECHO Algorithm:}
We prove that, under specified conditions, an echo chamber contains at least one node whose JHO value exceeds a theoretically derived threshold, providing a foundation for local seed-based detection.
We then design JECHO, a fast two-stage seed-based framework that combines JHO with score-based seed expansion for scalable ECD (Section~\ref{s:alg}).

\item \textbf{Empirical Validation:} Extensive experiments on both real-world and synthetic network data demonstrate that JECHO outperforms state-of-the-art methods in terms of both detection quality and runtime (Section~\ref{s:exp}).
\end{squishlist}


\section{Related Work}
\label{s:relatedwork}

\spara{Empirical Prevalence and Societal Impact.} Echo chambers have been empirically observed on platforms such as Twitter ($\mathbb{X}$)~\cite{nasim2022we,villa2021echo}, Reddit~\cite{impicciche2024comparing}, and YouTube~\cite{diaz2023disinformation}.
More broadly, a large body of work has examined polarization, opinion dynamics, selective exposure, and misinformation in online social networks, highlighting the systemic challenges that motivate the study of echo chambers, cf.~\cite{shirzadi2025stubborn,shirzadi2025opinion,del2016spreading,bakshy2015exposure}.
Echo chambers play an important role in polarized political discourse, the dissemination of misinformation, and ideological segregation during elections and crises~\cite{jiang2021mechanisms,amendola2024towards,minici2022cascade}. Extensive computational research indicates that algorithmic filtering and social influence exacerbate these phenomena, resulting in heightened opinion fragmentation and dismissal of evidence~\cite{sasahara2021social,chitra2020analyzing}.

\spara{Conceptual Characteristics of Echo Chambers.}
Echo chambers are commonly characterized by three properties:
\emph{internal opinion homogeneity}~\cite{del2016echo},
\emph{opinion extremism}~\cite{bright2018explaining,sunstein2002law},
and \emph{structural isolation}~\cite{cinelli2021echo,mahmoudi2024echo}.
The two \emph{opinion-related} characteristics capture how repeated exposure to aligned viewpoints reinforces shared beliefs and drives opinions toward more extreme positions within a group~\cite{calderon2019content,cota2019quantifying}.
The \emph{structural} characteristic refers to limited connectivity between the group and the rest of the network, resulting in topologically isolated communities that restrict exposure to opposing views~\cite{zhu2021influence}.
A recent survey provides a unified computational view of ideological isolation, covering echo chambers, filter bubbles, polarization, and selective exposure, and highlights the lack of consistent definitions and metrics across studies~\cite{wang2026ideological}.

\spara{Community-Detection-Based Methods.} Predominantly, computational approaches to ECD integrate community detection with opinion and content analysis. Standard methodologies initially segment the network using community detection methods, such as Louvain~\cite{blondel2008fast}, and subsequently measure intra-group consensus and inter-group divergence through opinions or embeddings~\cite{alatawi2023quantifying,ghafouri2024transformer}. 
Random-walk-based models exploit localized diffusion patterns within subgraphs to capture echo chamber effects~\cite{villa2021echo}. Hybrid approaches further incorporate signed networks or stance similarity, exemplified by SE-Echo in SEDA~\cite{perera2024quantifying} and embedding-based cohesion in EchoGAE~\cite{alatawi2023quantifying}. Although effective for medium-scale networks, these methods generally require global graph processing or repeated community re-optimization, thereby limiting the scalability on graphs with millions of nodes, which is common in practice. Recent work also proposed an unsupervised transformer-based metric that quantifies echo chamber effects using user embedding diversity and inter-echo chamber separability~\cite{ghafouri2024transformer}.

\spara{Single (and Local) Community Detection.}
Local graph partitioning and community detection aim to identify dense subgraphs using local procedures such as personalized PageRank~\cite{andersen2006local}, flow-based cuts~\cite{bae2017scalable}, heat-kernel methods~\cite{kloster2014heat}, and peeling~\cite{miyauchi2018finding}. Although these approaches scale well, they do not model opinion homogeneity or extremism, which are essential for ECD.
Recently, Chen et al.~\cite{chen2025qdisco} proposed the Q-DISCO framework to identify a cohesive subgraph with a relatively large (or small) average opinion.
Although Q-DISCO incorporates opinion polarity through a query opinion parameter, it only constrains the average opinion of the returned set and only optimizes the internal density rather than the structural isolation.
Consequently, the dense subgraphs identified often include moderately or weakly aligned nodes and fail to satisfy the required opinion extremism and structural isolation.

\spara{Dynamics and Diffusion Models.} Another research line studies the co-evolution of opinions and network structure. Minici et al.~\cite{minici2022cascade} infer latent communities and polarity from diffusion cascades, while Zhu et al.~\cite{zhu2021influence} incorporate echo chamber effects into influence maximization. Adaptive voter and temporal models further explain how rewiring and reinforcement drive polarization~\cite{timpanaro2024emergence,cau2024trends}. While these approaches elucidate echo chamber formation, they are not designed for scalable detection in large static networks. Existing methods either scale by relying on relaxed definitions that capture only partial properties or adopt richer definitions that do not scale.

\section{Problem Formulation and Characterization}
\label{s:prel}
We introduce a formal, unified definition of echo chambers that integrates the three core properties consistently identified in the literature:
\emph{internal opinion homogeneity}, \emph{opinion extremism}, and \emph{structural isolation}. Based on this definition, we formulate our ECD problem and prove the NP-hardness, establishing a fundamental barrier for exact computation.

\subsection{Notation and Definitions}
\label{ssec:notation}

A social network is represented as a directed graph $G = (V, E)$, where each node $v\in V$ corresponds to a user, and each edge $(u,v)\in E$ represents that $u$ affects $v$. 
Let $n = |V|$ and $m = |E|$ denote the number of nodes and edges, respectively. 
Undirected graphs can be represented as directed graphs by replacing each edge $\{u,v\}$ with two directed edges, $(u,v)$ and $(v,u)$. All undirected datasets in our experiments were processed accordingly.

For a node $v \in V$, the sets of outgoing and incoming neighbors are defined as
$\Gamma^+(v) = \{ u \in V \mid (v,u) \in E \}$ and
$\Gamma^-(v) = \{ u \in V \mid (u,v) \in E \}$, respectively.
Accordingly, the \emph{out-degree} and \emph{in-degree} of $v$ are
$\deg^+(v) = |\Gamma^+(v)|$ and $\deg^-(v) = |\Gamma^-(v)|$.
For any two subsets $V_1, V_2 \subseteq V$, we define the set of crossing edges as
$\partial(V_1, V_2) = \{ (u,v) \in E \mid u \in V_1,\ v \in V_2 \}$. 
When we need to make the underlying graph explicit, we write $\partial_G(V_1, V_2)$.

The structural isolation is defined as
$\mathcal{I}(S,V) = \frac{|\partial(S, V \setminus S)|}{|\partial(S,S)|}$ for $S\subset V$,
which measures the ratio of edges leaving $S$ to the number of edges that are internal to $S$. Also, we assume $|\partial(S,S)|>0$; otherwise $\mathcal{I}(S,V)$ is defined as $+\infty$.
Note that edges are counted only in the direction from $S$ to $V\setminus S$, which is consistent with the directionality of influence in social networks.
Furthermore, $\mathcal{I}(S, V)$ can be viewed as a conductance-style measure $\frac{|\partial(S, V\setminus S)|}{\sum_{v \in S}\deg^+(v)}$~\cite{chung1997spectral} that replaces the sum of degrees in standard conductance with the number of internal edges.
We intentionally normalize by the number of internal edges rather than the volume $\sum_{v\in S}\deg^+(v)$ because internal edges represent opinion reinforcement within the group, whereas outgoing edges represent information leakage; this ratio directly captures the balance between reinforcement and exposure that characterizes echo chambers. 
We assume that each user $v \in V$ is associated with an opinion value $o_v \in [-1,1]$ toward a topic: the larger $o_v$, the more favorable $v$ is toward the topic. Let $\mathbf{o} = (o_1,o_2,\ldots,o_n)$ be the opinion vector of users.
For a subset of nodes $S \subset V$, we define the mean and variance of the opinions as $\mathbb{E}(\mathbf{o}, S) = \frac{1}{|S|} \sum_{v \in S} o_v$ and $\sigma^2(\mathbf{o}, S) = \frac{1}{|S|} \sum_{v \in S} \left( o_v - \mathbb{E}(\mathbf{o}, S) \right)^2.$

\subsection{Definition of Echo Chambers}
\label{ssec:definition}
Echo chambers in social networks are commonly characterized by internal opinion homogeneity, limited exposure to opposing viewpoints, and opinions that are more extreme, i.e., further from the overall population mean, as reported in prior empirical and computational studies \cite{cinelli2021echo,bright2018explaining,mahmoudi2024echo}. Since no universally accepted mathematical definition exists, we adopt a principled formulation based on three complementary conditions: internal opinion homogeneity, opinion extremism, and structural isolation.

\begin{center}
\setlength{\fboxsep}{2pt}
\setlength{\fboxrule}{0.5pt}
\fcolorbox{blue!40!black}{blue!5!white}{%
  \parbox{0.92\columnwidth}{%
For a directed graph $G=(V,E)$ with an opinion vector $\mathbf{o}=(o_1,o_2,\dots, o_n)$ and non-negative real-valued thresholds $\theta_{\sigma^2},\theta_E,\theta_{\mathcal{I}}\ge 0$, a node set
$S\subset V$ is called a $(\theta_{\sigma^2},\theta_E,\theta_{\mathcal{I}})$-echo chamber, if the following three conditions are satisfied:
\begin{enumerate}
\item \textbf{Internal Opinion Homogeneity:}
$\sigma^2(\mathbf{o},S)\le \theta_{\sigma^2}$.
\item \textbf{Opinion Extremism:}
$\big|\mathbb{E}(\mathbf{o},V)-\mathbb{E}(\mathbf{o},S)\big|\ge \theta_E$.
\item \textbf{Structural Isolation:}
$\mathcal{I}(S,V)
\le \theta_{\mathcal{I}}.$
\end{enumerate}
  }%
}
\end{center}

The variance measure $\sigma^2(\mathbf{o}, S)$ captures \emph{internal opinion homogeneity} by quantifying how tightly concentrated the opinions within $S$ are around their mean, with lower values indicating stronger internal agreement.
The extremism measure $|\mathbb{E}(\mathbf{o}, V)-\mathbb{E}(\mathbf{o}, S)|$ captures \emph{opinion extremism} by quantifying how far the average opinion of $S$ deviates from the global mean opinion, where larger values indicate more extreme ideological positioning.
The measure $\mathcal{I}(S, V)$ represents \emph{structural isolation} and quantifies the ratio of edges leaving the set to edges remaining within it. In our formulation for directed graphs, this captures how opinions flow outwards versus being reinforced internally, with lower values indicating stronger echo chamber isolation.
Together, these three conditions capture complementary aspects of echo chambers, ensuring that the identified groups are internally aligned, structurally isolated, and ideologically distinct from the broader population. Our formulation is guided by interpretability and  consistency with prior studies~\cite{mahmoudi2024echo,cinelli2021echo,bright2018explaining}.

\subsection{Echo Chamber Detection (ECD) Problem}
\label{ssec:problem-statement}

Although echo chambers may resemble dense or well-separated clusters,
ECD fundamentally differs from graph clustering because it simultaneously enforces both opinion homogeneity and extremism.
We formulate the ECD problem as follows:

\begin{center}
\setlength{\fboxsep}{2pt}
\setlength{\fboxrule}{0.5pt}
\fcolorbox{blue!40!black}{blue!5!white}{%
  \parbox{0.92\columnwidth}{%
\textbf{Input:}
A directed graph $G=(V,E)$, an opinion vector 

\noindent
$\mathbf{o}=(o_1,\dots,o_n)$, and two thresholds $\theta_{\sigma^2}$ and $\theta_E$.

\textbf{Output:}
A set $S\subset V$ that minimizes $\mathcal{I}(S,V)$ subject to
\[
\sigma^2(\mathbf{o}, S)\le \theta_{\sigma^2}
\quad\text{and}\quad
|\mathbb{E}(\mathbf{o},V)-\mathbb{E}(\mathbf{o},S)| \ge \theta_E .
\]
  }%
}
\end{center}

This formulation naturally captures our goal of finding the most structurally isolated group (i.e., the lowest $\mathcal{I}(S, V)$) that simultaneously maintains internal opinion homogeneity and opinion extremism.
Any solution $S$ returned by our ECD problem satisfies the echo chamber definition in Section~\ref{ssec:definition}, with $\theta_{\mathcal{I}}$ instantiated as the achieved structural isolation value $\mathcal{I}(S,V)$.

\subsection{Computational Complexity Analysis}
\label{ssec:hardness}

We prove that the ECD problem is NP-hard by constructing a polynomial-time reduction from the conductance minimization problem, which is known to be NP-hard~\cite{raghavendra2012reductions}. 
In this problem, given an undirected graph $G=(V,E)$, we are asked to find $S \subseteq V$ with $|S| \le |V|/2$ that minimizes the conductance $\frac{|\partial(S,V \setminus S)|}{\sum_{v\in S}\text{deg}(v)}$, where $\text{deg}(v)$ is the degree of $v\in V$. 

\begin{theorem}
\label{thm:hardness}
The ECD problem is NP-hard. 
\end{theorem}

\begin{proof}
Let $G=(V, E)$ be an instance of the conductance minimization problem. 
We make an instance $(G'=(V',E'),\mathbf{o},\theta_{\sigma^2},\theta_E)$ of the ECD problem as follows: 
The vertex set $V'$ is defined as $V'=V_\mathrm{copy}\cup V^+ \cup V^-$, where $V_\mathrm{copy}$ is the copy of the original vertex set $V$ of $G$, and $V^+$ and $V^-$ are disjoint sets, each of which is of size $\lfloor |V|/2 \rfloor$. 
The edge set $E'$ is a directed counterpart of the original edge set $E$ of $G$, i.e., $E'=\{(u,v),(v,u)\mid \{u,v\}\in E\}$. 

Note therefore that the nodes in $V^+$ and $V^-$ are all isolated. 
The opinion vector $\mathbf{o}$ is defined as $o_v=0$ for all $v\in V_\mathrm{copy}$, $o_v=1$ for all $v\in V^+$, and $o_v=-1$ for all $v\in V^-$, resulting in $\mathbb{E}(\mathbf{o},V')=0$. 
The thresholds $\theta_{\sigma^2}$ and $\theta_E$ are set as $\theta_E=1/2$ and $\theta_{\sigma^2}=+\infty$. 

Then, we see that for any $\texttt{val}>0$, the optimal value to the instance $G$ of the conductance minimization problem is less than or equal to $\texttt{val}$ if and only if the optimal value to the instance $(G'=(V', E'),\mathbf{o},\theta_{\sigma^2},\theta_E)$ to the ECD problem is less than or equal to $\frac{\texttt{val}}{1-\texttt{val}}$. 
We first prove ``$\Rightarrow$.'' 
By the assumption, there exists $S\subset V$ with $|S|\leq |V|/2$ that has a conductance less than or equal to $\texttt{val}$. 
Let $S_\mathrm{copy}$ be the corresponding vertex subset in $V'$. 
By simple calculation, we have 
\begin{align*}
\frac{|\partial_G(S,V\setminus S)|}{\sum_{v\in S}\mathrm{deg}(v)}\leq \texttt{val} \Leftrightarrow \frac{|\partial_{G'}(S_\mathrm{copy},V\setminus S_\mathrm{copy})|}{|\partial_{G'}(S_\mathrm{copy},S_\mathrm{copy})|}\leq \frac{\texttt{val}}{1-\texttt{val}}.
\end{align*}
As adding any isolated node to a subset does not affect the structural isolation value of the subset, $S_\mathrm{copy}\cup V^+$ and $S_\mathrm{copy}\cup V^-$ have the same structural isolation value as $S_\mathrm{copy}$. Recalling that $|S|\leq |V|/2$ (and thus $|S|\leq \lfloor |V|/2\rfloor$) and $|V^+|=|V^-|=\lfloor |V|/2\rfloor$, we have $\mathbb{E}(\mathbf{o},S_\mathrm{copy}\cup V^+)\geq 1/2$ and $\mathbb{E}(\mathbf{o},S_\mathrm{copy}\cup V^-)\leq -1/2$, meaning that 
$S_\mathrm{copy}\cup V^+$ and $S_\mathrm{copy}\cup V^-$ are feasible for the ECD problem. 

Next, we prove ``$\Leftarrow$.'' 
By the assumption, there exists $S\subset V'$ with $|\mathbb{E}(\mathbf{o}, V')-\mathbb{E}(\mathbf{o}, S)|\geq \theta_E=1/2$ that has the structural isolation value less than or equal to $\frac{\texttt{val}}{1-\texttt{val}}$. 
By the opinion extremism constraint, it is easy to see that $|S\cap V_\mathrm{copy}|\leq \lfloor |V|/2 \rfloor \leq |V|/2$. From the above equivalence, we see that the vertex subset corresponding to $S\cap V_\mathrm{copy}$ in $V$ has the conductance less than or equal to $\texttt{val}$. 
\end{proof}

This NP-hardness result rules out the possibility of an efficient exact algorithm, unless $\text{P} = \text{NP}$, motivating the design of efficient practical methods for large-scale instances of the ECD problem.

\section{The JECHO Algorithm}
\label{s:alg}

We propose a theory-guided \emph{seed selection and expansion algorithm} that identifies echo chambers through local exploration around seed nodes rather than a global search. Our algorithm, \textbf{JECHO}, involves two stages:

\begin{enumerate}[leftmargin=*, itemsep=0pt, parsep=3pt, topsep=3pt]
 \item \textbf{Seed selection}, which identifies a small set of nodes likely to belong to echo chambers.
 \item \textbf{Seed expansion}, which grows each seed node into a full echo chamber by optimizing structural isolation while preserving internal opinion homogeneity and opinion extremism. 
\end{enumerate}
This design is motivated by Theorem~\ref{thm:jho-existence} (described in Section~\ref{ssec:seed-theory}), which guarantees that every echo chamber contains at least one node with high local homophily and strong structural cohesion.

\subsection{Jaccard-Based Homophily (JHO)}
\label{ssec:jho}

We first define a primary concept to be used in our algorithm design, which we refer to as Jaccard-based homophily (JHO). For any two finite sets $A$ and $B$, the Jaccard index is defined as
$\mathrm{JI}(A,B)=\frac{|A\cap B|}{|A\cup B|}.$
In our setting, we quantify the \emph{structural similarity} between nodes by applying the Jaccard index to their outgoing neighborhoods. For nodes $w,v\in V$, we define
\[
\mathrm{JI}(w,v):=\mathrm{JI}(\Gamma^+(w),\Gamma^+(v))
= \frac{|\Gamma^+(w)\cap\Gamma^+(v)|}{|\Gamma^+(w)\cup\Gamma^+(v)|}.
\]
To jointly capture opinion agreement and structural cohesion, for any node $w\in V$, we define the JHO score as follows:
\begin{equation}
\mathrm{JHO}(w):=
\frac{\sum_{v\in\Gamma^+(w)} \mathrm{JI}(w,v)
\left(1-\frac{(o_w-o_v)^2}{4}\right)}
{\sum_{v\in\Gamma^+(w)} \mathrm{JI}(w,v)}.
\label{eq:jho}
\end{equation}
Here, the factor $4$ normalizes the squared opinion difference because $o_v\in[-1,1]$, ensuring $\frac{(o_w-o_v)^2}{4}\in[0,1]$ and hence $\mathrm{JHO}(w)\in[0,1]$ whenever the denominator is nonzero. 
We define $\mathrm{JHO}(w)=0$ if the denominator in Eq.~\eqref{eq:jho} is zero, as this indicates the absence of local structural overlap.
Intuitively, $\mathrm{JHO}(w)$ is large when $w$ is embedded in a locally cohesive neighborhood (high Jaccard similarity) whose members are opinion-aligned with $w$.
For convenience, we also define the complementary heterogeneity score by 
\[
\mathrm{JHet}(w):= 1-\mathrm{JHO}(w)=
\frac{\sum_{v\in\Gamma^+(w)} \mathrm{JI}(w,v)\frac{(o_w-o_v)^2}{4}}
{\sum_{v\in\Gamma^+(w)} \mathrm{JI}(w,v)}.
\]

\subsection{Theoretical Basis for Seed Selection}
\label{ssec:seed-theory}

Here, we provide a theoretical basis for seed selection using $\mathrm{JHO}$ scores. Specifically, we show that any vertex subset $S$ that satisfies some conditions (which echo chambers are likely to satisfy in practice) contains at least one node with a large $\mathrm{JHO}$ score.

\begin{theorem}
\label{thm:jho-existence}
Let $S \subset V$ be a vertex subset that satisfies the following conditions: 
\begin{enumerate}[leftmargin=*, itemsep=3pt, parsep=3pt, topsep=3pt]
\item Opinion alignment: all nodes $v\in S$ with $o_v\neq 0$ have opinions with the same sign (i.e., either positive or negative). 
\item Internal cohesion:
let $\alpha = \frac{\sum_{(u,v)\in\partial(S,S)} \mathrm{JI}(u,v)}{|\partial(S,S)|}$ and $\alpha' = \frac{\sum_{(u,v)\in\partial(S,V\setminus S)} \mathrm{JI}(u,v)}
{|\partial(S,V\setminus S)|}$ 
denote the average structural similarities over internal and outgoing edges of $S$,
respectively. Then, it holds that $\alpha > \alpha'$.
\end{enumerate}
Then there exists a node $w \in S$ such that $\mathrm{JHO}(w) > \frac{3}{4(1+\mathcal{I}(S,V))}.$
\end{theorem}

\begin{proof}
For each $v\in S$, we define the sum of structural similarities between $v$ and nodes in $S$, and between $v$ and nodes outside $S$ as follows: 
$
\mathrm{in}(v)=\sum_{u\in \Gamma^+(v)\cap S}\mathrm{JI}(u,v)\quad 
\text{and}\quad 
\mathrm{out}(v)=\sum_{u\in \Gamma^+(v)\setminus S}\mathrm{JI}(u,v). 
$
By the internal-cohesion assumption (2) in the statement, we have 
\begin{align*}
\frac{\sum_{v\in S}\mathrm{in}(v)}{|\partial(S,S)|} 
> \frac{\sum_{v\in S}\mathrm{out}(v)}{|\partial(S,V\setminus S)|},\quad 
\text{i.e.,}\quad 
\frac{|\partial(S,V\setminus S)|}{|\partial(S,S)|}
> \frac{\sum_{v\in S}\mathrm{out}(v)}{\sum_{v\in S}\mathrm{in}(v)}, 
\end{align*}
meaning that there exists $w\in S$ that satisfies $\mathcal{I}(S,V)>\frac{\mathrm{out}(w)}{\mathrm{in}(w)}$. 
Then, using the opinion-alignment assumption (1), we lower bound $\mathrm{JHO}(w)$ by defining $\psi(w,u):=\mathrm{JI}(w,u)\left(1-\frac{(o_w-o_u)^2}{4}\right)$. Since $\mathrm{out}(w) < \mathcal{I}(S,V)\mathrm{in}(w)$, we have

\begin{align*}
\mathrm{JHO}(w)
&= \frac{\sum_{u\in\Gamma^+(w)\cap S}\psi(w,u)
+\sum_{u\in\Gamma^+(w)\setminus S}\psi(w,u)}
{\mathrm{in}(w)+\mathrm{out}(w)} \\
&\ge \frac{\frac{3}{4}\mathrm{in}(w)}{\mathrm{in}(w)+\mathrm{out}(w)} 
> \frac{3}{4(1+\mathcal{I}(S,V))}. 
\end{align*}
This completes the proof. 
\end{proof}

%
%
%
%

This theorem shows that, under the stated opinion-alignment and internal-cohesion assumptions, an echo chamber contains at least one node with a high JHO score. Accordingly, we select nodes with the highest JHO scores as candidate seeds.

\subsection{Algorithm}
\label{ssec:jecho}

We present our algorithm, which we refer to as JECHO, in Algorithm~\ref{alg:seed_selection}.
We first pre-filter nodes to obtain a subset $\mathcal{B}$ containing only the most extreme nodes by $|o_v|$, retaining a fixed proportion of the highest-magnitude opinions to satisfy opinion extremism requirements. For each node in $\mathcal{B}$, the algorithm computes its \textbf{JHO} score.
While $\mathcal{B}\neq \emptyset$ and we do not reach the $k_\mathrm{max}$ echo chamber candidates, the algorithm repeats the following process: select a seed node $v^*\in \mathcal{B}$, expand it into an echo chamber candidate using SSE, introduced below, and remove the candidate's nodes from $\mathcal{B}$. Finally, it returns the best echo chamber found during expansion.

\begin{algorithm}[!htpb]
\caption{JECHO}
\label{alg:seed_selection}
\KwIn{Graph $G=(V,E)$, opinions $\mathbf{o}$, thresholds $\theta_{\sigma^2}$, $\theta_E$}
\KwOut{Echo chamber $\mathcal{C}^*$}
\BlankLine

$n \gets |V|$

$k_{\max} \gets \lfloor \rho_s n \rfloor$ \tcp*[r]{$\rho_s$: seed ratio}

$\mathcal{B} \gets$ top $\lfloor \rho_e n \rfloor$ nodes in $V$ by $|o_v|$ \tcp*[r]{$\rho_e$: extremeness ratio}

\ForEach{$v \in \mathcal{B}$}{
$s_v \gets \texttt{JHO}(v)$
}

$k \gets 0$, $\mathcal{C}^* \gets \emptyset$, $\mathcal{I}^* \gets \infty$

\While{$\mathcal{B} \neq \emptyset \land k < k_{\max}$}{
$v^* \gets$ node in $\mathcal{B}$ with the highest $s_v$

$k \gets k + 1$

$(\mathcal{C},\mathcal{I}) \gets \texttt{SSE}(G,v^*,\mathbf{o},\theta_{\sigma^2},\theta_E)$ 

$\mathcal{B} \gets \mathcal{B} \setminus \mathcal{C}$

\If{$\mathcal{I} < \mathcal{I}^*$}{
$\mathcal{I}^* \gets \mathcal{I}$,
$\mathcal{C}^* \gets \mathcal{C}$
}
}
\Return $\mathcal{C}^*$
\end{algorithm}

Now, we describe how the \textbf{SSE} expands a given seed node $v^*$ into an echo chamber and how nodes are prioritized during the expansion process.
Starting from the seed node $v^*$, the \textbf{SSE} greedily expands $v^*$ into an echo chamber by iteratively adding neighboring nodes while enforcing internal opinion homogeneity, opinion extremism, and minimizing structural isolation. At each iteration, SSE selects the highest-priority candidate node, adds it to the current echo chamber, and updates the set of candidate nodes accordingly. Throughout the expansion, SSE tracks the best echo chamber encountered based on structural isolation and returns the resulting echo chamber together with its isolation value.

SSE maintains a priority ordering over candidate nodes. Each node $v$ is scored based on its structural embedding in the current echo chamber, combining its internal-to-external connectivity ratio with its Jaccard-based homophily score $\mathrm{JHO}(v)$. Priorities are updated incrementally as the chamber grows, only for nodes whose adjacency changes, ensuring efficient expansion. A complete specification is provided in Algorithm~\ref{alg:expansion}.
Consistent with our ECD problem formulation and to enable a uniform comparison across methods, we report the best echo chamber identified in each run. JECHO can also return multiple candidate echo chambers by retaining the results obtained from additional seeds.

\begin{algorithm}[!htpb]
\caption{Score-Based Seed Expansion (SSE)}
\label{alg:expansion}
\KwIn{Graph $G=(V,E)$, seed $v^*$, opinions $\mathbf{o}$, thresholds $\theta_{\sigma^2},\theta_E$}
\KwOut{Best echo chamber $\mathcal{C}^*$ and its $\mathcal{I}^*$}
\BlankLine

\textbf{Init: }
$\mathcal{C}\gets \emptyset$, 
$\mathcal{C}^*\gets \emptyset$, 
$\mathcal{I}^*\gets \infty$,
$\mathcal{H}\gets \{(v^*,0)\}$,
$\mathcal{A}\gets \{v^*\}$
\BlankLine

$\mathrm{stats}\gets \emptyset$ \tcp*[r]{cached connectivity statistics}

\While{$\mathcal{H}\neq\emptyset $}{
$u\gets \mathcal{H}.\text{pop}()$

$\mathcal{C}\gets \mathcal{C}\cup\{u\}$

$(V_S,\Delta_E,\mathcal{I})\gets$ compute opinion homogeneity, extremism, and structural isolation of $\mathcal{C}$

\If{$|\mathcal{C}|\ge 10 \land V_S\le\theta_{\sigma^2} \land \Delta_E\ge\theta_E \land \mathcal{I}\le\mathcal{I}^*$}{
$\mathcal{C}^*\gets \mathcal{C}$,
$\mathcal{I}^*\gets \mathcal{I}$
}

$\mathcal{N}\gets \{v \in \Gamma^+(u) \mid \mathbf{o}[v]o_{v^*}>0,\ v\notin \mathcal{A}\}$

$\mathcal{A}\gets \mathcal{A}\cup \mathcal{N}$

$\mathcal{N}_{\text{affected}} \gets 
\{ w \in \mathcal{H}.\text{keys()} \mid
\Gamma^+(w)\cap \mathcal{N} \neq \emptyset \}$

\ForEach{$v \in \mathcal{N} \cup \mathcal{N}_{\text{affected}}$}{
    \If{$v \notin \mathrm{stats}$}{
        $\mathrm{in} \gets |\Gamma^+(v) \cap \mathcal{C}|$, $\mathrm{out} \gets |\Gamma^+(v)| - \mathrm{in}$
    }
    \Else{
        $(\mathrm{in},\mathrm{out}) \gets \mathrm{stats}[v]$
        
        $\mathrm{in} \gets \mathrm{in} + 1$, $\mathrm{out} \gets \mathrm{out} - 1$
    }

    $\mathrm{stats}[v] \gets (\mathrm{in},\mathrm{out})$

$\text{score} \gets \texttt{JHO}(v)$

$\mathrm{priority} \gets 
\begin{cases}
-\mathrm{in}, & \mathrm{out}=0 \\
-(\mathrm{in}/\mathrm{out}) \cdot \text{score} \cdot |\mathbf{o}[v]|, & \text{otherwise}
\end{cases}$

\eIf{$v \in \mathcal{N}$}{
$\mathcal{H}.\text{push}(v,\mathrm{priority})$
}{
$\mathcal{H}.\text{update}(v,\mathrm{priority})$
}
}
}
\Return $\mathcal{C}^*,\mathcal{I}^*$
\end{algorithm}

\subsection{Runtime Analysis}
\label{ssec:complexity}
Let $n=|V|$, $m=|E|$, and $k_{\max}$ be the maximum number of seeds expanded
(typically $k_{\max}=O(\rho_s n)$).
Assuming Jaccard similarities are precomputed/cached, all JHO scores can be computed in $O(m)$ time.
Each SSE execution updates the required metrics incrementally and maintains a priority queue, yielding
$O((m+n)\log n)$ time per run.
Thus, JECHO runs in $O(k_{\max}(m+n)\log n)$ time (details in Appendix~\ref{apdx:timecomplexity}).

\section{Experimental Evaluation}
\label{s:exp}

We evaluate JECHO on real-world and synthetic network structures under a unified experimental protocol with multiple opinion-generation settings, complemented by two established datasets containing real opinion labels.
Throughout this section, modularity denotes the directed modularity of Leicht and Newman~\cite{leicht2008community}, which measures how strongly edges are concentrated within communities in a directed graph.
All implementations are available at \href{https://github.com/aSafarpoor/JECHO-CIKM}{https://github.com/aSafarpoor/JECHO-CIKM}.

\subsection{Datasets and Opinion Assignment}
\spara{Real-world networks.}
JECHO is evaluated on several real-world networks from the SNAP repository~\cite{snapnets}, including Facebook (4.0K nodes), LastFM (7.6K), Git (37.7K), Twitter (81.3K), and Pokec (1.63M), along with the opinion-based labeled datasets Twitter-RW (549 nodes) and Reddit-RW (557 nodes) datasets~\cite{de2014learning}. 
These networks span scales from hundreds to over one million nodes and exhibit diverse densities and community structures, providing a challenging testbed for the ECD problem.

\spara{Labeled real-world datasets.}
A key limitation in evaluating echo chamber detection methods is the scarcity of real-world networks that provide both graph structure and reliable opinion labels. The Twitter-RW and Reddit-RW datasets introduced earlier partially address this requirement.
However, these datasets are relatively small and correspond to subgraphs of larger platforms; thus, they may not fully capture the structural complexity of large-scale networks. Nevertheless, they provide a practical setting for evaluating the model behavior under real opinion signals.

\spara{Synthetic networks.}
In addition to real-world networks, we generated synthetic graphs using the Stochastic Block Model (SBM)~\cite{clauset2004finding}. Intra- and inter-block edge probabilities were adjusted to achieve target modularity levels, enabling controlled evaluation across different community strengths.
We consider SBM graphs with varying sizes (SBM\mbox{-}5K, SBM\mbox{-}10K, SBM\mbox{-}50K, and SBM\mbox{-}100K), where the suffix denotes the number of nodes. 

\spara{Opinion generation and variants.}
Node opinions were provided for Twitter-RW and Reddit-RW and rescaled to $[-1,1]$. 
For all other datasets, in the absence of ground-truth opinion labels, opinions were generated synthetically using a multistage pipeline designed to reflect realistic opinion formation processes. The nodes were first partitioned using a streaming neighborhood procedure, with each partition assigned a latent mean opinion drawn from a Gaussian distribution. Individual node opinions were then sampled from a Gaussian distribution centered on the mean of their partition and subsequently refined using the Friedkin-Johnsen (FJ) model~\cite{friedkin1990social}.
Starting from the FJ opinions, we constructed two additional variants. The Filtering variant applies a homophily-based filtering step inspired by Springsteen et al.~\cite{springsteen2024algorithmic} to increase local opinion homogeneity, while the Extreme variant fixes the opinions of a small fraction of nodes at $\pm 1$, encouraging stronger polarization.
This yields three opinion variants: \textbf{FJ} (baseline equilibrium opinions), \textbf{Filtering} (enhanced local homogeneity), and \textbf{Extreme} (highly polarized opinions). In all setups, Gaussian noise was added, and opinions were normalized to $[-1,1]$.

\subsection{Compared Methods}
\label{ssec:Compatedmethods}

JECHO is compared with representative baseline methods spanning seed-based, community-based, and query-centric ECD paradigms. For seed selection, we consider degree, node-level homophily, an entropy-based score adapted from BeECD~\cite{wang2023beecd} and AttriRank~\cite{hsu2017unsupervised}. For seed expansion, we use AttriPPR, which is a personalized PageRank variant that combines structural transitions with attribute (opinion) similarity through a weighted mixture.

Given a seed node, AttriPPR produces a ranking over nodes, which are then added sequentially to form candidate echo chambers while enforcing constraints on homogeneity, extremeness, and structural isolation. The combination of JHO seed selection with AttriPPR expansion is denoted as JHO\&APPR. The procedure is summarized in Algorithm~\ref{alg:attrippr_expansion}.
We also evaluated two state-of-the-art community-based ECD methods. SEDA~\cite{perera2024quantifying} was adapted to an opinion-only setting using a signed adjacency matrix, whereas EchoGAE~\cite{alatawi2023quantifying} was adapted by applying Louvain clustering followed by echo chamber selection. 
Baselines require certain adaptations to align with our ECD formulation, though their core objectives and structures are preserved. As these methods target related but distinct objectives, we evaluate them under the proposed formulation, where structural isolation, JECHO's own objective, is the primary metric.
Adaptation details are in Appendix~\ref{APX:baselines}.

\begin{algorithm}[!htpb]
\caption{AttriPPR-Based Expansion}
\label{alg:attrippr_expansion}

\KwIn{Graph $G=(V,E)$, seed $v^*$, opinions $\mathbf{o}$, thresholds $\theta_{\sigma^2},\theta_E$}
\KwOut{Best echo chamber $\mathcal{C}^*$ and its $\mathcal{I}^*$}

$\pi \gets$ AttriPPR scores for $v^*$ based on $G=(V,E)$ and $\mathbf{o}$

$V' \gets$ nodes sorted by $\pi$ (descending)

$\mathcal{C} \gets \{v^*\}$, $\mathcal{C}^* \gets \{v^*\}$, $\mathcal{I}^* \gets \infty$

\For{$u \in V'$}{
    $\mathcal{C} \gets \mathcal{C} \cup \{u\}$

    $(V_S,\Delta_E,\mathcal{I})\gets$ compute opinion homogeneity, extremism, and structural isolation of $\mathcal{C}$

    \If{$|\mathcal{C}| \ge 10 \land V_S \le \theta_{\sigma^2} \land \Delta_E \ge \theta_E \land \mathcal{I} \le \mathcal{I}^*$}{
   $\mathcal{C}^* \gets \mathcal{C}$,
   $\mathcal{I}^* \gets \mathcal{I}$
    }
}

\Return $\mathcal{C}^*$, $\mathcal{I}^*$
\end{algorithm}

Finally, we implemented Q-Peeling \cite{chen2025qdisco}, a query-centric densest subgraph method. In our setting, the query $q \in \{-1,+1\}$ specifies the target opinion polarity, where $q=+1$ (resp. $q=-1$) corresponds to communities aligned with positive (resp. negative) opinions with opinion close to $\pm 1$. For each dataset, Q-Peeling was run with both $q=+1$ and $q=-1$, and the resulting communities were evaluated using the same ECD constraints.
We additionally considered singleton echo chambers as a lower-bound baseline, used only when no non-trivial echo chamber satisfying the ECD constraints was found. A node $v$ is deemed feasible if its opinion deviates from the global mean by at least $\theta_E$. For a singleton ${v}$, structural isolation is defined as $\mathcal{I}({v},V)=\deg^+(v)$, and we report the feasible singleton minimizing $\deg^+(v)$.

\subsection{Experimental Procedure}

For each dataset and opinion variant, all methods were evaluated over $10$ independent runs with different random opinion initializations. For Twitter-RW and Reddit-RW, which contained fixed opinions, experiments were conducted once. In each run, for the detected echo chamber $S$, we measured structural isolation $\mathcal{I}(S, V)$, runtime, and success rate, defined as the fraction of runs producing a valid echo chamber of at least $10$ nodes. Unless otherwise stated, the results are averaged over the runs.

All experiments were conducted on a single workstation with 32\,GB RAM. The SBM parameters, opinion generation settings, and detection thresholds were fixed to appropriate values across the datasets. For real-world graphs, we used $\theta_{\sigma^2}=0.075$ and $\theta_E=0.3$, while for SBMs $\theta_{\sigma^2}=0.05$ and $\theta_E=0.3$. The seed ratio was set to $0.005$ for SBMs and small networks and $0.001$ for real-world networks. A minimum echo chamber size of 10 nodes was enforced. The full hyperparameter settings are provided in Appendix~\ref{APX:setup}.
To ensure scalability and a fair comparison, all methods were implemented with memory-efficient data structures (please see Appendix~\ref{APX:heapsizebound}).

The default threshold values were chosen to avoid overly permissive settings that would admit loosely structured or weakly aligned groups while ensuring meaningful levels of opinion homogeneity and extremism consistent with the notion of echo chambers. These values were selected based on empirical validation of smaller instances to balance feasibility and selectivity.

\subsection{Overall Performance Comparison}
\label{subsec:Overall-Performance-Comparison}
Table~\ref{tab:result_combined} compares JECHO with all baseline methods on both real-world and synthetic networks. The results are reported in terms of structural isolation, with lower values indicating stronger isolation.
For the Pokec dataset, which is the largest used network in this paper, with the FJ-based opinion assignment, JECHO achieved $\mathcal{I}=0.42$ with an average runtime of $6{,}429$ seconds, while JHO\&APPR obtained $\mathcal{I}=1.22$ with an average runtime of $15{,}155$ seconds. Q-Peeling did not identify any echo chambers, and the remaining methods failed to return results within the imposed time limit.

\begin{table*}[!htpb]
\centering
\caption{ECD results on different networks. Lower $\mathcal{I}$ indicates stronger structural isolation. Results are reported for three opinion settings: FJ, Filt (Filtering), and Ext (Extreme). Timeouts or invalid runs use the singleton baseline (Section~\ref{ssec:Compatedmethods}); entries with no valid results are shown in \textcolor{violet}{violet}. Best and second-best results are shown in bold and underlined, respectively (excluding violet). 
Success rates (valid echo chambers with at least 10 nodes found within the time limit) are reported in Section~\ref{subsec:Overall-Performance-Comparison}.}
\label{tab:result_combined}
\small
\setlength{\tabcolsep}{1.1pt}
\begin{tabular}{|cc|
ccc|ccc|ccc|ccc|
ccc|ccc|ccc|ccc|}
\hline

\timerow
 & &
\multicolumn{3}{c}{Facebook} &
\multicolumn{3}{c}{LastFM} &
\multicolumn{3}{c}{Git} &
\multicolumn{3}{c|}{Twitter} &
\multicolumn{3}{c}{SBM\mbox{-}5K} &
\multicolumn{3}{c}{SBM\mbox{-}10K} &
\multicolumn{3}{c}{SBM\mbox{-}50K} &
\multicolumn{3}{c|}{SBM\mbox{-}100K} \\

Method & Metric
& FJ & Filt & Ext
& FJ & Filt & Ext
& FJ & Filt & Ext
& FJ & Filt & Ext
& FJ & Filt & Ext
& FJ & Filt & Ext
& FJ & Filt & Ext
& FJ & Filt & Ext \\
\hline 

JECHO
& $\mathcal{I}$
& \underline{0.51} & \underline{0.09} & \underline{0.06} & \textbf{0.46} & \textbf{0.15} & \textbf{0.20} & \textbf{0.94} & \textbf{1.00} & \textbf{1.00} & \underline{0.15} & \textbf{0.05} & \textbf{0.10}
& \textbf{0.73} & \textbf{0.42} & \textbf{0.54} & \textbf{1.03 }& 0.48 & \textbf{0.55} & \textbf{1.00} & \textbf{0.95 }& \textbf{0.98} & \textbf{1.10} & \textbf{0.96} & \textbf{1.00} \\
\timerow & \scriptsize{Time}
& \scriptsize{0.38} & \scriptsize{0.38} & \scriptsize{0.42} & \scriptsize{0.32} & \scriptsize{0.38} & \scriptsize{0.40} & \scriptsize{11.9} & \scriptsize{10.8} & \scriptsize{11.1} & \scriptsize{23.5} & \scriptsize{30.7} & \scriptsize{30.6}
& \scriptsize{1.3} & \scriptsize{1.4} & \scriptsize{1.5} & \scriptsize{5.1} & \scriptsize{5.6} & \scriptsize{6.1} & \scriptsize{44.4} & \scriptsize{44.2} & \scriptsize{44.6} & \scriptsize{185} & \scriptsize{184} & \scriptsize{187} \\
\hline

JHO\&APPR
& $\mathcal{I}$
& 0.73 & 0.21 & 0.31 & \underline{0.80} & 0.23 & \underline{0.28} & \textcolor{violet}{1.00} & \textcolor{violet}{1.00} & \textcolor{violet}{1.00} & \textbf{0.13} & \underline{0.11} & \underline{0.14}
& \underline{0.87} & 0.53 & \underline{0.94} & \underline{1.08} & \underline{0.42} & \textbf{0.55} & \textbf{1.00} & \underline{1.00} & \underline{1.00} & \textbf{1.10} & \underline{1.00} & \textbf{1.00} \\
\timerow & \scriptsize{Time}
& \scriptsize{0.35 } & \scriptsize{0.37 } & \scriptsize{0.36 } & \scriptsize{0.21 } & \scriptsize{0.20 } & \scriptsize{0.22 } & \scriptsize{4.2 } & \scriptsize{4.3 } & \scriptsize{4.2 } & \scriptsize{37.8 } & \scriptsize{37.5 } & \scriptsize{36.1
} & \scriptsize{0.26 } & \scriptsize{0.26 } & \scriptsize{0.27 } & \scriptsize{1.0 } & \scriptsize{1.1 } & \scriptsize{1.2 } & \scriptsize{15.8 } & \scriptsize{16.1 } & \scriptsize{16.0 } & \scriptsize{101 } & \scriptsize{105 } & \scriptsize{109} \\
\hline

SEDA
& $\mathcal{I}$
& \textbf{0.26 }& \textbf{0.03} & \textbf{0.05} & \underline{0.80} & \underline{0.17} & 0.55 & \underline{1.00} & \textcolor{violet}{1.00} & \textbf{1.0} & 0.21 & \underline{0.11} & 0.40
& 1.00 & 0.89 & 1.4 & 1.37 & 1.62 & 3.13 & \textcolor{violet}{1.00} & \textcolor{violet}{1.00} & \textcolor{violet}{1.00} & \textbf{1.10} & \textcolor{violet}{1.20} & \textcolor{violet}{1.10} \\
\timerow & \scriptsize{Time}
& \scriptsize{12.1 } & \scriptsize{10.5 } & \scriptsize{12.4 } & \scriptsize{8.6 } & \scriptsize{7.5 } & \scriptsize{8.1 } & \scriptsize{306 } & \scriptsize{2158 } & \scriptsize{860 } & \scriptsize{1036 } & \scriptsize{926 } & \scriptsize{980
} & \scriptsize{13.7 } & \scriptsize{14.5 } & \scriptsize{14.4 } & \scriptsize{49.6 } & \scriptsize{56.2 } & \scriptsize{52.1 } & \scriptsize{242 } & \scriptsize{245 } & \scriptsize{190 } & \scriptsize{1102 } & \scriptsize{1269 } & \scriptsize{1129} \\
\hline

EchoGAE
& $\mathcal{I}$
& 0.61 & 0.13 & 0.31 & 0.94 & 0.79 & 0.81 & \textcolor{violet}{1.00} & \textcolor{violet}{1.00} & \textcolor{violet}{1.00} & 0.71 & \underline{0.11} & 0.81
& \underline{0.87} & \underline{0.50} & 1.16 & \textcolor{violet}{2.80} & \textbf{0.34} & \underline{0.87} & \textcolor{violet}{1.00} & \textcolor{violet}{1.00} & \textcolor{violet}{1.00} & \textcolor{violet}{1.10} & \textcolor{violet}{1.20} & \textcolor{violet}{1.10} \\
\timerow & \scriptsize{Time}
& \scriptsize{3.5 } & \scriptsize{3.5 } & \scriptsize{3.4 } & \scriptsize{13.4 } & \scriptsize{13.3 } & \scriptsize{13.2 } & \scriptsize{26.6 } & \scriptsize{26.7 } & \scriptsize{27.0 } & \scriptsize{256 } & \scriptsize{253 } & \scriptsize{249
} & \scriptsize{5.2 } & \scriptsize{5.2 } & \scriptsize{5.2 } & \scriptsize{19.6 } & \scriptsize{19.6 } & \scriptsize{19.3 } & \scriptsize{655 } & \scriptsize{563 } & \scriptsize{564 } & \scriptsize{1919 } & \scriptsize{1904 } & \scriptsize{1920} \\
\hline

Q-Peeling
& $\mathcal{I}$
& 0.91 & 0.66 & 0.82 & \textcolor{violet}{1.00} & \textcolor{violet}{1.00} & \textcolor{violet}{1.00} &
\textcolor{violet}{1.00} & \textcolor{violet}{1.00} & \textcolor{violet}{1.00} &
\textcolor{violet}{1.00} & \textcolor{violet}{1.00} & \textcolor{violet}{1.00}
& \textcolor{violet}{1.00} & \textcolor{violet}{1.00} & \textcolor{violet}{1.40}
& \textcolor{violet}{2.80} & \textcolor{violet}{3.00} & \textcolor{violet}{3.30}
& \textcolor{violet}{1.00} & \textcolor{violet}{1.00} & \textcolor{violet}{1.00}
& \textcolor{violet}{1.10} & \textcolor{violet}{1.20} & \textcolor{violet}{1.10} \\
\timerow & \scriptsize{Time}
& \scriptsize{6.5} & \scriptsize{6.2 } & \scriptsize{6.2 } & \scriptsize{5.4 } & \scriptsize{2.5 } & \scriptsize{2.5 } & \scriptsize{142.4 } & \scriptsize{39.2 } & \scriptsize{39.6 } & \scriptsize{595 } & \scriptsize{240 } & \scriptsize{231
} & \scriptsize{4.1 } & \scriptsize{4.0 } & \scriptsize{3.9 } & \scriptsize{13.8 } & \scriptsize{13.4 } & \scriptsize{13.0 } & \scriptsize{54.3 } & \scriptsize{54.1 } & \scriptsize{52.2 } & \scriptsize{266 } & \scriptsize{265 } & \scriptsize{260} \\
\hline
\end{tabular}
\end{table*}

In addition, results for the labeled datasets (Reddit-RW and Twitter-RW) are presented in Figure~\ref{fig:RealWorldSensitivity}. On Reddit-RW, only JECHO and SEDA identify echo chambers under limited threshold settings. On Twitter-RW, JECHO and JHO\&APPR detect echo chambers under a broader range, whereas the other methods work only under relaxed configurations with low opinion extremization. These results highlight the difficulty of the task on labeled real-world data and further demonstrate JECHO’s robustness.

\spara{Detection Quality.}
Across different datasets (Table~\ref{tab:result_combined}), JECHO achieves the
lowest structural isolation values in most cases, indicating more
sharply isolated echo chambers.
For example, on Twitter with extreme opinions, JECHO attains $\mathcal{I}=0.1$, compared with $0.4$ for SEDA, whereas EchoGAE fails to return a valid result.
JECHO also exhibited the lowest failure rate, and JHO\&APPR was typically the second-best.
On the Facebook network, which has high modularity ($0.84$), community-based methods such as SEDA and EchoGAE perform reasonably well, reflecting that modular networks are favorable to clustering-based approaches.
However, on less modular networks, such as Git and Pokec, these baselines frequently fail to meet the ECD constraints, whereas JECHO continues to detect valid echo chambers with low structural isolation.

\spara{Scalability.}
All experiments were conducted within a 12-hour time limit per configuration.
JECHO successfully completed all real-world experiments, including Pokec with $1.63$ million nodes. In contrast, SEDA timed out on Pokec. Large-scale datasets also impose memory limitations, which we mitigated using a memory-aware implementation.
Although all methods have polynomial worst-case complexity, SEDA and EchoGAE operate on the full graph in each iteration, whereas JECHO performs only $k_{\max}\ll n$ local expansions determined by the seed ratio.
This locality yields a substantially lower effective runtime on large graphs.

\begin{figure}[t]
    \centering
    \includegraphics[width=0.75\linewidth]{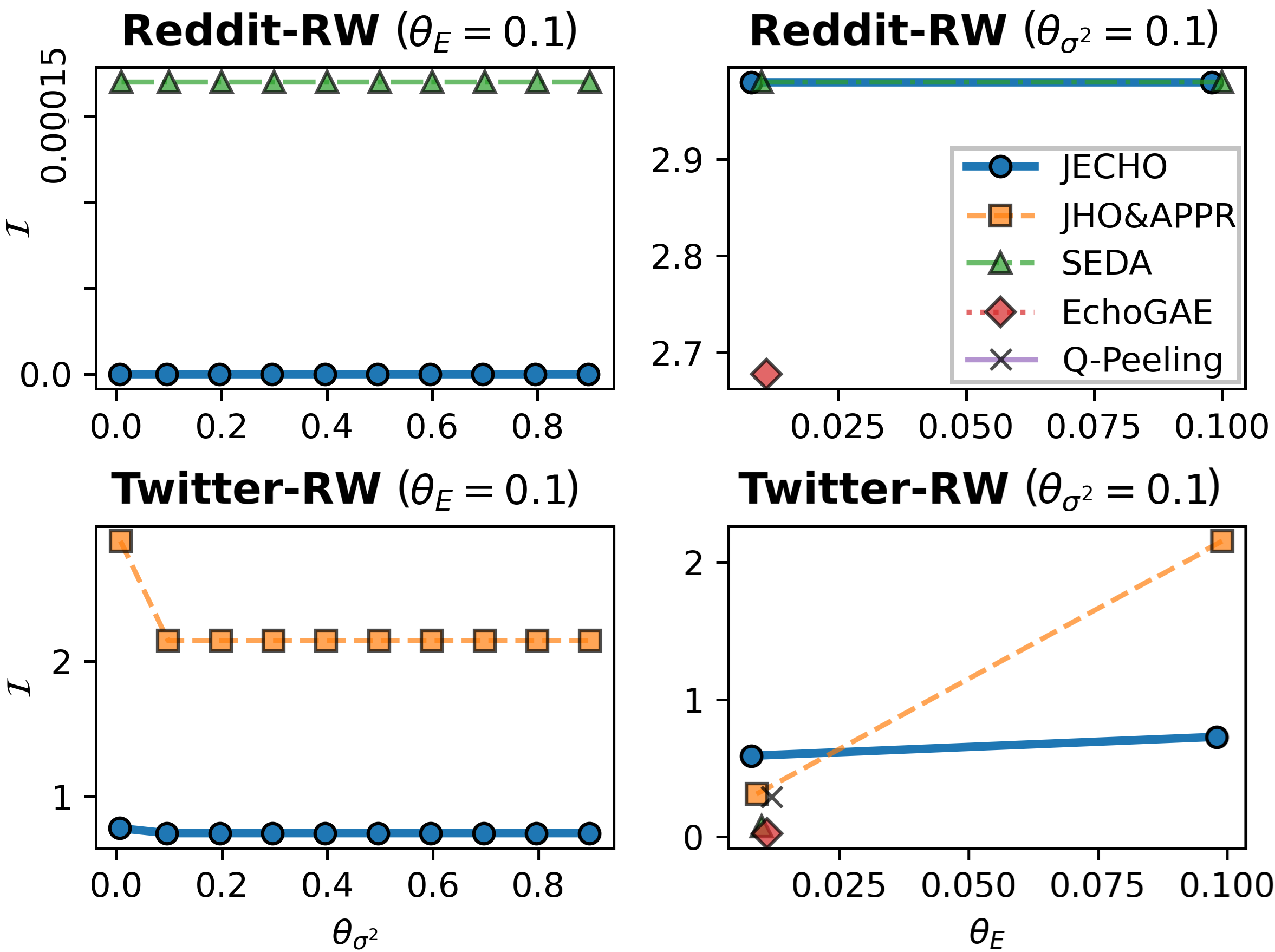}
    \caption{Results on labeled real-world datasets. In the left column, $\theta_{\sigma^2}$ varies while $\theta_E$ is fixed at 0.1. In the right column, $\theta_E$ varies while $\theta_{\sigma^2}$ is fixed at 0.1.}
    \Description{Results on the Reddit-RW and Twitter-RW labeled datasets under varying echo chamber detection thresholds. The left column varies the opinion homogeneity threshold while fixing the extremism threshold, and the right column varies the extremism threshold while fixing the homogeneity threshold. The plots compare the structural isolation achieved by JECHO and the baseline methods across these settings.}
    \label{fig:RealWorldSensitivity}
\end{figure}

\spara{Success Rate.}
Over 10 independent runs, JECHO returned valid echo chambers in $\mathrm{94}\%$ of cases, compared with $\mathrm{70}\%$ for JHO\&APPR, $\mathrm{50}\%$ for SEDA, and $\mathrm{29}\%$ for EchoGAE. In addition, Q-Peeling succeeded in $\mathrm{3}\%$ of runs.
Failures are most common in networks with low modularity or diffuse opinion distributions, where the global community structure is weak.
The success rate of JECHO compared with other models demonstrates substantially greater reliability in practical deployment.

\subsection{Ablation and Case Studies}

We first perform an ablation study to isolate the effects of JECHO’s JHO-based seed selection and SSE-based expansion by replacing each component individually while keeping the remainder of the pipeline fixed. This allows us to assess the contribution of each design choice to both detection quality and efficiency. We then present controlled and qualitative case studies to further illustrate the behavior of JECHO under varying structural conditions and to provide intuitive insights into the detected echo chambers.

\subsubsection{Seed Selection Ablation Study}
We first fix the expansion strategy to SSE and vary only the seed selection function. Figure~\ref{fig:score_seed_comparison} compares different seed scoring methods on the LastFM and SBM\mbox{-}10K dataset, including degree-, homophily-, entropy-, and AttriRank-based selection. In most cases, JECHO seed scoring yields lower structural isolation values. These findings empirically support the theoretical basis for seed selection in Section~\ref{ssec:seed-theory}, which predicts that nodes with high JHO scores are more likely to belong to echo chambers.
Additional results are reported in Appendix~\ref{apx:additional_results}.

\begin{figure}[!htpb]
 \centering
 \includegraphics[width=0.9\linewidth]{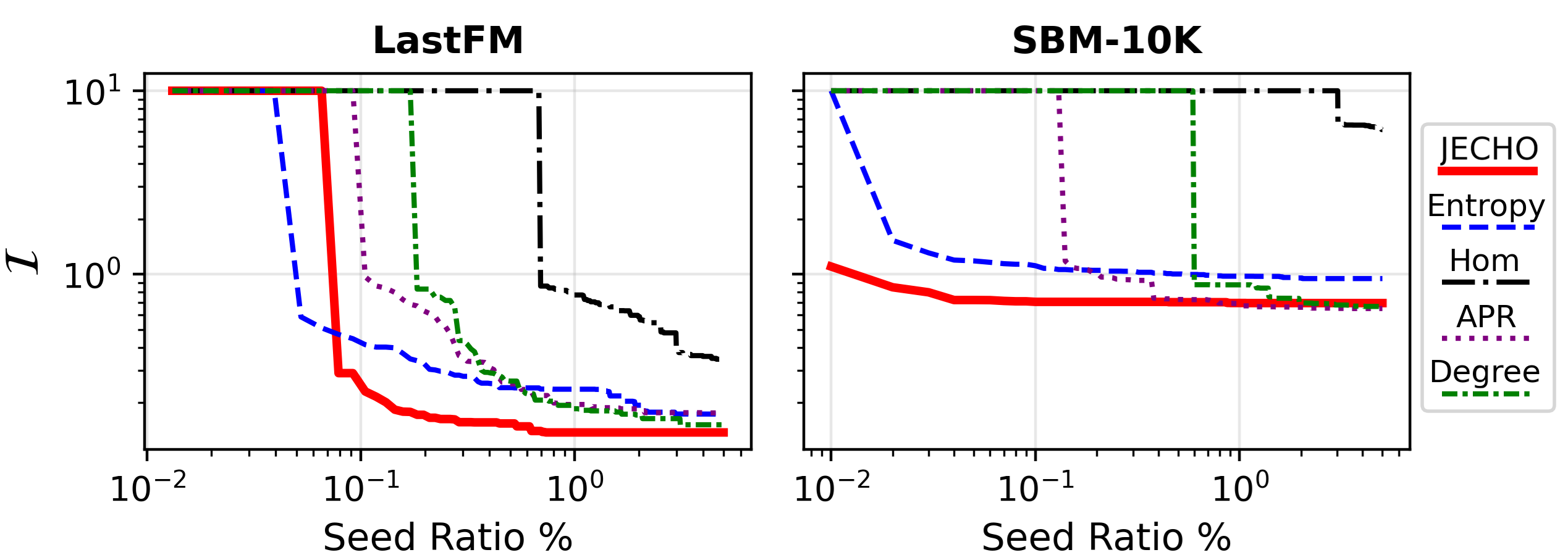}
 \caption{Seed scoring ablation on the LastFM and SBM\mbox{-}10K datasets. The structural isolation obtained by SSE is shown as a function of the seed ratio for different seed scoring functions. Lower values indicate stronger echo chamber quality.}
 \Description{Line plot showing structural isolation versus seed ratio for different seed scoring functions, including JHO, entropy, homophily, degree, and AttriRank.}
\label{fig:score_seed_comparison}
\end{figure}

\subsubsection{Seed Expansion Ablation Study}
Next, we fix the seed selection method and compare the seed expansion part. Figure~\ref{fig:expansion_score_comparison} compares SSE (used in JECHO) and AttriPPR on the LastFM dataset under different node scoring functions. In most cases, SSE achieves stronger structural isolation, while AttriPPR is faster, consistent with Table~\ref{tab:result_combined}. Overall, SSE offers a better balance between isolation quality and efficiency and remains feasible on larger graphs. Additional results are provided in Appendix~\ref{apx:additional_results}.

\begin{figure}[t]
 \centering
 \includegraphics[width=0.9\linewidth]{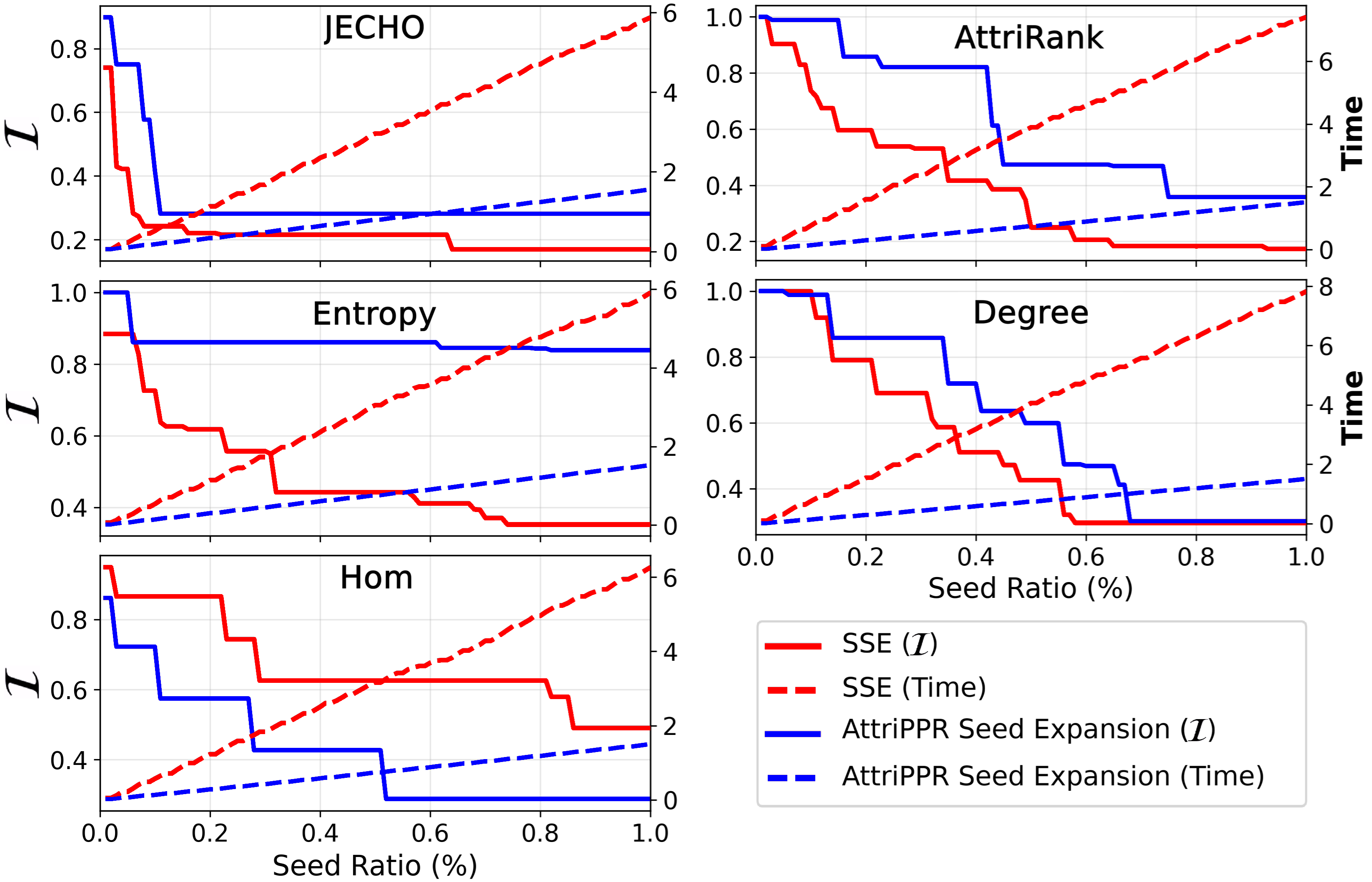}
 \caption{Comparison of seed expansion strategies on the LastFM dataset. Solid lines denote structural isolation $(\mathcal{I})$, while dashed lines indicate runtime, for SSE and AttriPPR under different seed scoring functions.}
 \Description{Plots comparing SSE and AttriPPR seed expansion, showing structural isolation with solid lines and runtime with dashed lines under different seed scoring functions.}
\label{fig:expansion_score_comparison}
\end{figure}

\subsubsection{Case Study on Modularity Effects}
\label{sec:ModularityCase}
This case study examines the effect of graph modularity on ECD performance using an SBM with $20\text{K}$ nodes, where the modularity parameter is varied in $[0.1, 0.9]$. As shown in Figure~\ref{fig:modularity_case_Study}, increasing modularity generally simplifies the ECD task for all methods by strengthening community structure.
Higher modularity consistently improves structural isolation across all methods. At low modularity levels, JECHO achieves the strongest isolation with the lowest runtime, followed by JHO\&APPR, demonstrating its effectiveness in detecting subtle echo chambers across different modularity levels.

\begin{figure}[t]
 \centering
 \includegraphics[width=0.95\linewidth]{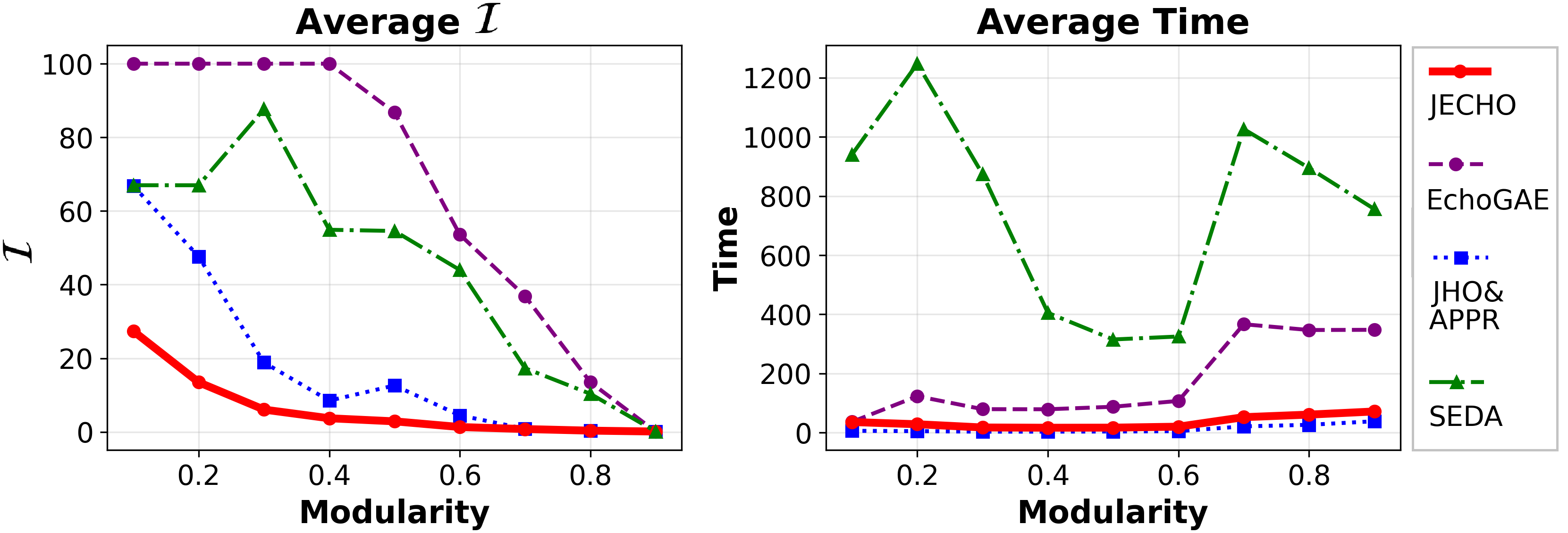}
 \caption{Performance comparison of different ECD methods across varying modularity levels in SBM-generated networks.}
 \Description{Plots showing average structural isolation and average runtime of different ECD models across increasing modularity levels in SBM-generated networks.}
\label{fig:modularity_case_Study}
\end{figure}

\subsubsection{Sensitivity to Metrics and Thresholds}

To assess the robustness of JECHO with respect to both modeling choices and parameter settings, we conduct a sensitivity analysis over alternative metric formulations and threshold configurations.
For \textbf{metric sensitivity}, we adopt a controlled replacement strategy in which each component is substituted with a standard alternative while keeping the remaining components fixed. Specifically, structural isolation is alternatively measured using conductance,
$
\frac{|\partial(S, V \setminus S)|}{|\partial(S, V)|},
$
and opinion homogeneity is measured using mean absolute deviation,
$
\frac{1}{|S|} \sum_{v \in S} |o_v - \mathbb{E}(\mathbf{o}, S)|,
$
where $S \subseteq V$ denotes a candidate echo chamber.
For \textbf{threshold sensitivity}, we vary the thresholds $\theta_{\sigma^2}$ and $\theta_E$ over a broad range, capturing different levels of strictness in what is identified as an echo chamber.

Results in Figure~\ref{fig:SensetivityStudy} exhibit largely similar qualitative patterns across both alternative metric formulations and the original formulation (ours). As the thresholds vary, the overall trends remain stable. This suggests that the effectiveness of JECHO is not tied to a specific choice of homogeneity or isolation measure, but rather to the underlying principles captured by their joint modeling.

\begin{figure}[!htbp]
    \centering
    \includegraphics[width=0.85\linewidth]{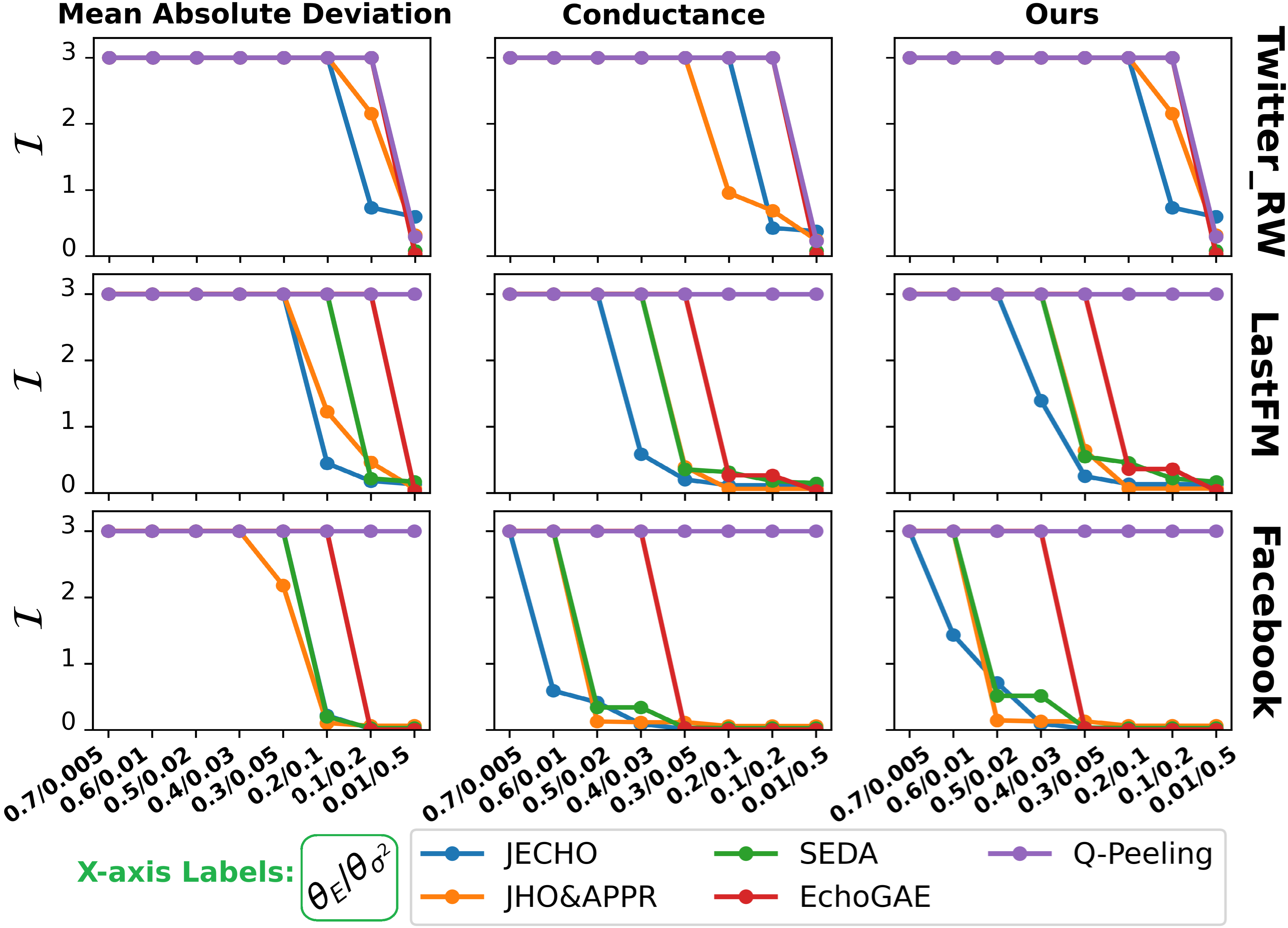}
    \caption{Sensitivity analysis. Columns show alternative homogeneity, alternative structural isolation, and the original formulation. Thresholds progress from stricter to more relaxed settings, illustrating method performance across threshold changes. As shown in the hint, the x-axis labels represent combinations of $\theta_{\sigma^2}$ and $\theta_E$.}
    \Description{Sensitivity analysis across Twitter-RW, LastFM, and Facebook. Columns correspond to alternative opinion homogeneity using mean absolute deviation, alternative structural isolation using conductance, and the original formulation. Each plot compares JECHO and the baseline methods as combinations of the homogeneity and extremism thresholds progress from stricter to more relaxed settings.}
    \label{fig:SensetivityStudy}
\end{figure}

\subsubsection{Qualitative Visualization on Real-World Data}

To complement the quantitative and controlled case studies, we provide a qualitative visualization of a detected echo chamber in a real-world network. As shown in Figure~\ref{fig:qualitative_detection}, the identified subgraph exhibits strong opinion alignment (node colors) and clear structural isolation, which illustrate the key properties captured by our formulation.

\begin{figure}[!htpb]
 \centering
 \includegraphics[width=0.7\linewidth]{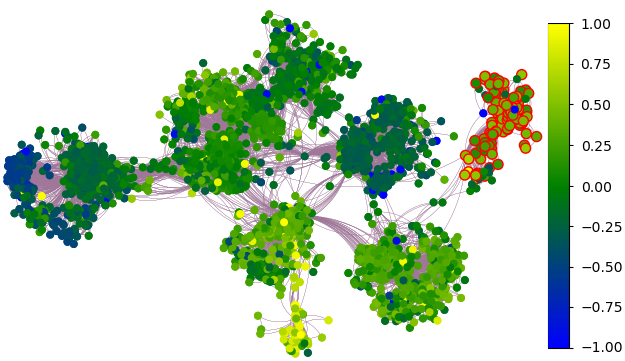}
 \caption{An echo chamber detected in the Facebook dataset under filtering-based opinion initialization. Node colors indicate opinion in $[-1,1]$ (blue: negative, yellow: positive); red boundaries mark echo chamber nodes.}
 \Description{Network visualization of an echo chamber detected in the Facebook dataset, where node color encodes opinion from negative to positive and red boundaries indicate nodes belonging to the detected echo chamber. }
\label{fig:qualitative_detection}
\end{figure}

\subsubsection{Illustrative Toy Case Studies}
Illustrative toy case studies are provided to highlight qualitative differences between models.

\spara{Case 1: Embedded Extreme Cluster.}
A graph with 100 nodes is constructed, containing a small dense cluster of 20 extreme-opinion nodes ($[0.5,1.0]$) weakly connected to 80 moderate-opinion nodes ($[-0.25,0.25]$).
The extreme cluster is internally dense and weakly connected to the moderate cluster.
The goal is to recover the small extreme group despite the surrounding moderate structure.
SEDA fails to detect a valid echo chamber, while JECHO and JHO\&APPR successfully identify the extreme cluster.

\spara{Case 2: Boundary Trade-off in a Star Graph.} 
A star-shaped graph with seven arms and a smooth opinion gradient from the center ($+1$) to the periphery ($-1$) is considered.
This setting highlights the trade-off between opinion extremism and structural isolation.
Although EchoGAE detects an echo chamber, its isolation quality is poorer.
JECHO and JHO\&APPR recover more meaningful boundaries by balancing opinion alignment and structural separation.
Figure~\ref{fig:toy_cases} illustrates both cases.

\begin{figure}[!htpb]
 \centering
 \includegraphics[width=0.9\linewidth]{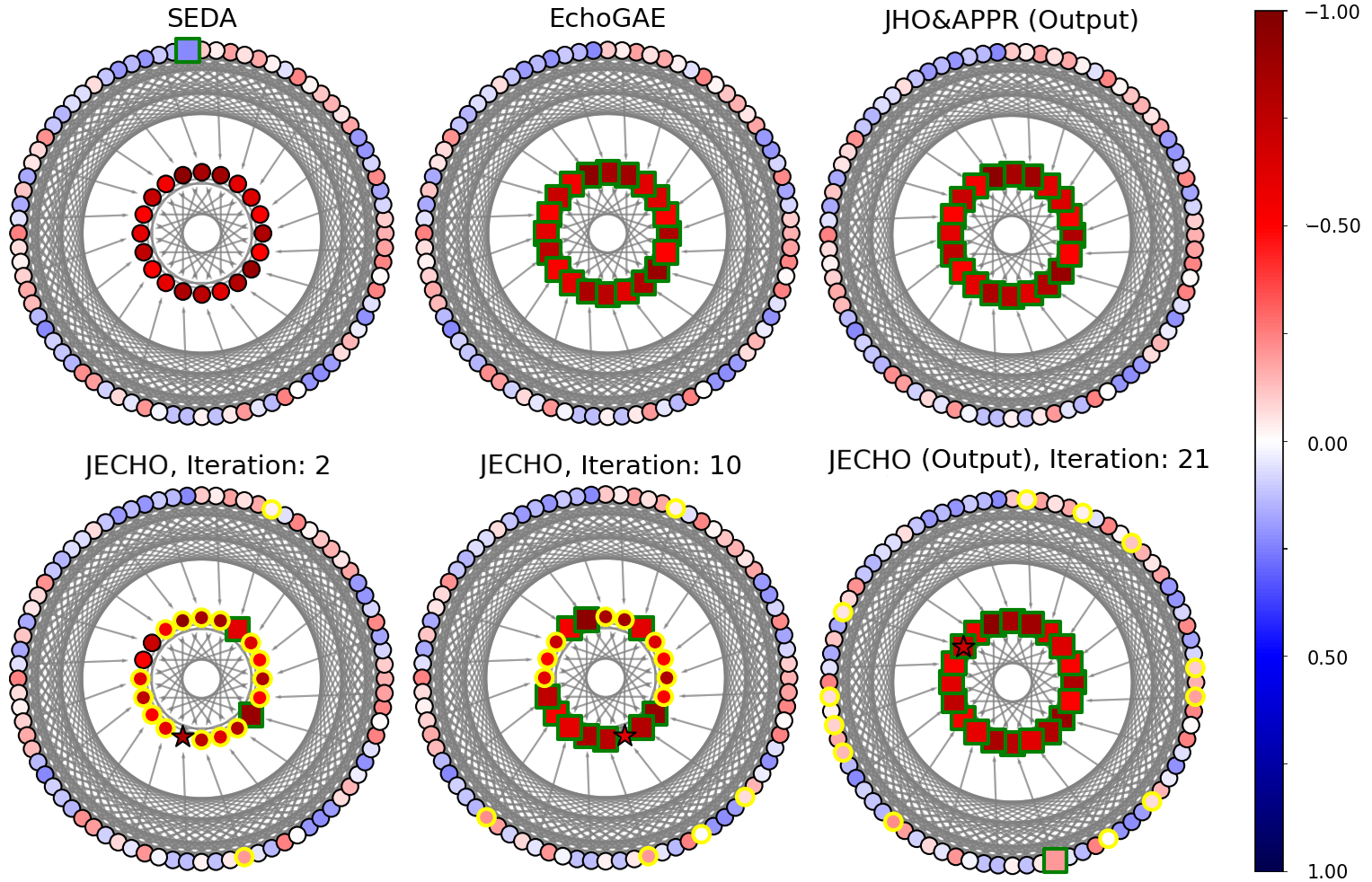}\\
 \vspace{0.50em}
 \includegraphics[width=0.9\linewidth]{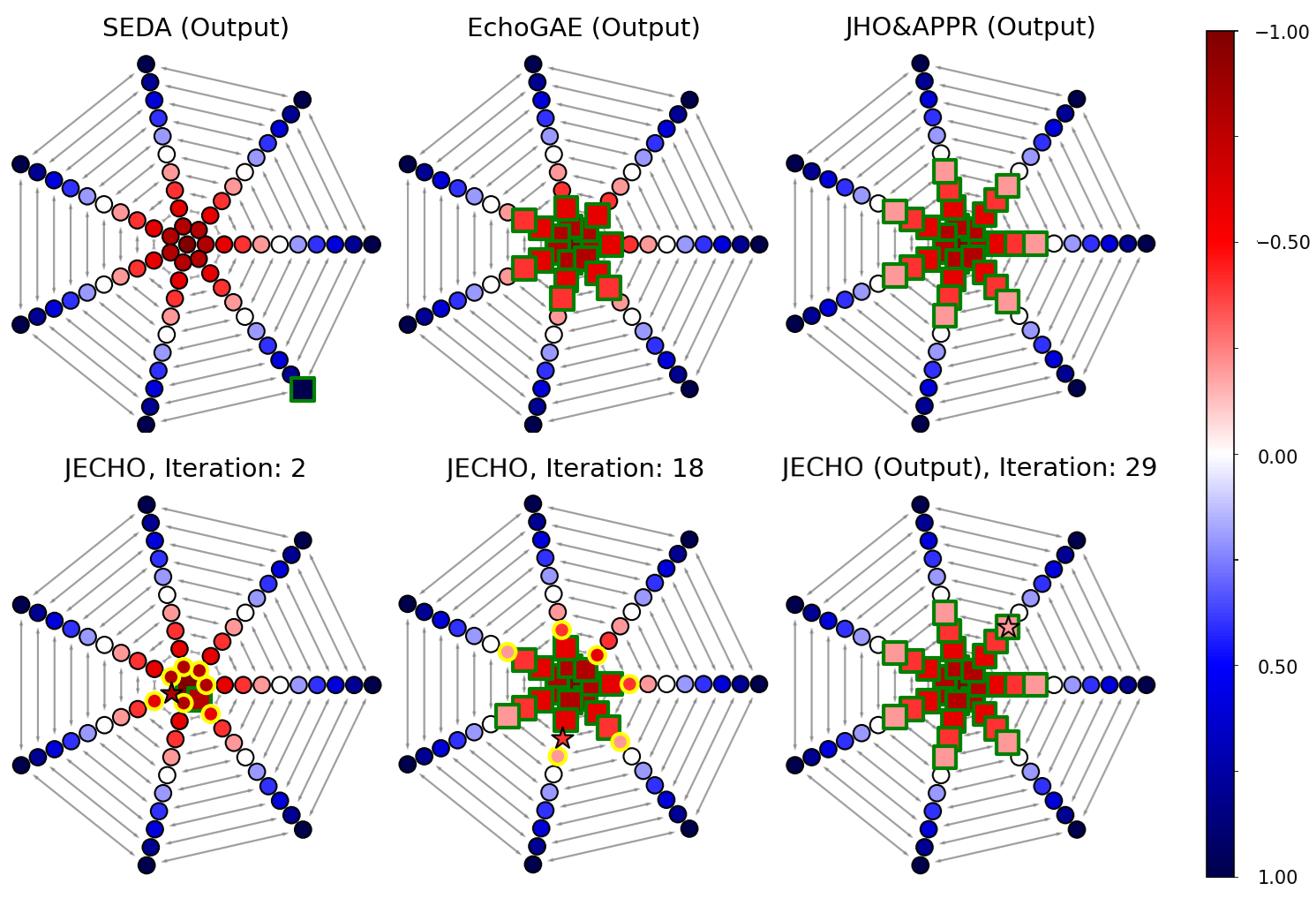}
 \caption{
 Qualitative case study.
Top: Case~1 (embedded extreme cluster).
Bottom: Case~2 (star graph with opinion gradient).
The first row shows the baseline outputs, and the subsequent rows illustrate JECHO’s iterative expansion.
The node color indicates the opinion.
Green square: selected node; star: current expansion node; yellow border: candidate nodes.}
\Description{Qualitative case studies illustrating two synthetic graphs: an embedded extreme cluster and a star graph with opinion gradient. Rows show baseline outputs and successive JECHO expansion steps, with node color indicating opinion and markers highlighting selected and candidate nodes.}
 \label{fig:toy_cases}
\end{figure}

\section{Conclusions, Limitations, and Future Work}
\label{s:conclusion}
In this work, we addressed the lack of algorithmic rigor in echo chamber detection (ECD) by establishing a unified definition that enforces internal opinion homogeneity, opinion extremism, and structural isolation. We proved that our ECD problem is NP-hard, revealing a fundamental barrier to exact computation. 
To overcome this, we derived a theoretical lower bound on the Jaccard-based homophily of nodes within opinion-aligned and internally cohesive echo chambers.
This theoretical insight grounded the design of JECHO, a local expansion framework that circumvents global graph partitioning.
Our extensive evaluation on datasets scaling to millions of nodes shows that JECHO reduces runtime by orders of magnitude compared to community-detection baselines while identifying more sharply isolated groups. Moreover, JECHO remains effective in low-modularity settings where traditional methods often fail.

\smallskip
\noindent \textbf{Limitations}. Despite its strengths, our proposal has limitations. First, our model relies on scalar opinion values, which simplifies complex ideological positions into a single dimension. While effective for studying polarization, this does not capture the multifaceted nature of belief systems in rich text data. 
Second, JECHO is designed for static network snapshots and does not model the temporal dynamics of echo chamber formation or dissolution.

\smallskip
\noindent \textbf{Future Work.} List of main directions for future investigation:
\begin{squishlist}
    \item Temporal Dynamics: Extending JECHO to dynamic graphs to track the lifecycle (formation, solidification, and dissolution) of echo chambers over time.
    
    \item Rich Content Integration: Incorporating semantic embeddings (e.g., from LLMs) alongside scalar opinion values to detect ``topic-specific'' echo chambers.
 
    \item Multidimensional Ideological Space: Incorporating \emph{ideological embeddings} (e.g., \cite{Monti0AB21}) to detect echo chambers in a multifaceted, multidimensional ideological space.
  
    \item Mitigation Strategies: Targeting the identified boundary nodes from the Seed Expansion phase for intervention strategies, such as bridge-building recommendations to reduce polarization.

    \item Theoretical Guarantees: Developing approximation guarantees and performance bounds for the proposed framework.
\end{squishlist}

\appendix

\section{JECHO Time Complexity Analysis}
\label{apdx:timecomplexity}

For a node $w$, \textit{score computation} requires iterating over its out-neighbors: 
$O\!\left(
\sum_{v\in\Gamma^+(w)} T\bigl(\mathrm{JI}(w,v)\bigr)
\right).$
If Jaccard indices are stored on edges, this reduces to
$O(\deg^+(w))$, and the total score computation cost over all nodes
is $T_{\mathrm{score}}=O(m)$.
If node scores are precomputed and the seed bank is maintained as a max-priority queue, seed selection based on Algorithm~\ref{alg:seed_selection} runs in
$T_{\mathrm{selection}}= O(n\log n + k_{\max}\cdot T_{\mathrm{SSE}})$.
Assuming $\texttt{JHO}(v)$ is an $O(1)$ lookup as scores are computed earlier, the SSE
(Algorithm~\ref{alg:expansion}) runs in $T_{\mathrm{SSE}} = O((m+n)\log n)$.
Finally, \textit{JECHO} runs in
$
T_{\mathrm{JHO}}=
T_{\mathrm{score}} + T_{\mathrm{selection}}=
O\!\left(m\right) + O\!\left(n\log n + k_{\max}\cdot T_{\mathrm{SSE}}\right)=
O\!\left( m + n\log n + k_{\max}\cdot (m+n)\log n  \right)
=
O\!\left(
k_{\max}\cdot (m+n)\log n
\right).
$

\section{Baseline Adaptations}
\label{APX:baselines}

Two state-of-the-art ECD methods were adapted to our opinion-only setting.
\emph{SEDA} was implemented using a signed adjacency matrix $W_{ij}=o_i o_j$, with agreement and disagreement matrices $A=\max(0,W)$ and $D=\min(0,W)$. 
It iteratively reassigned nodes to maximize signed modularity, rewarding internal agreement and penalizing external disagreement. Communities with positive SE scores were treated as echo chambers.
\emph{EchoGAE} was adapted by first detecting communities using Louvain clustering. 
It then computed, for each community, a cohesion--separation score from within- and across-community opinion distances. The highest-scoring community was selected as the echo chamber. Embedding components were omitted because only opinion values were available.

\section{Hyperparameters and Thresholds}
\label{APX:setup}

For SBM graphs with $n$ nodes, the number of blocks is set to $\lceil \sqrt{n} \rceil$ with target modularity $0.7$, and the number of opinion partitions is set to $2.5$ times the number of blocks. For real-world graphs, the number of partitions is set to $2\lceil \sqrt{n} \rceil$. Partition means are sampled from $\mathcal{N}(0,0.2)$ for SBMs and $\mathcal{N}(0,0.3)$ for real networks. Five filtering iterations with five sampled neighbors per node are used. The extreme-node and noise ratios are set to $1\%$ and $0.03$, respectively. Detection thresholds are $\theta_{\sigma^2}=0.05$ and $\theta_E=0.3$ for SBMs, and $\theta_{\sigma^2}=0.075$ and $\theta_E=0.3$ for real networks. The seed ratio is set to $0.005$ for SBMs and $0.001$ for real graphs, with an extreme pre-filtering ratio of $0.1$ and a minimum echo chamber size of 10.

\section{Bounded-Heap Optimization for Efficient SSE}
\label{APX:heapsizebound}

To improve scalability, the SSE implementation uses a bounded priority queue of maximum size $K$ (set to 500), storing only the top-$K$ highest-scoring candidate nodes. A new candidate is inserted if the heap is not full; otherwise, it replaces the current minimum only if its score is higher.
This optimization preserves the SSE ranking logic while substantially reducing memory usage and runtime by avoiding large candidate sets. It is purely an implementation-level optimization for efficient execution on large graphs.

\section{Different Score Computation Models} 
\label{apx:additional_results}
Figures~\ref{fig:comparision-expansions} and~\ref{fig:score_seed_comparison_all2} compare seed expansion and seed scoring strategies across different datasets.

\begin{figure}[!htbp]
 \centering
 \includegraphics[width=0.99\linewidth]{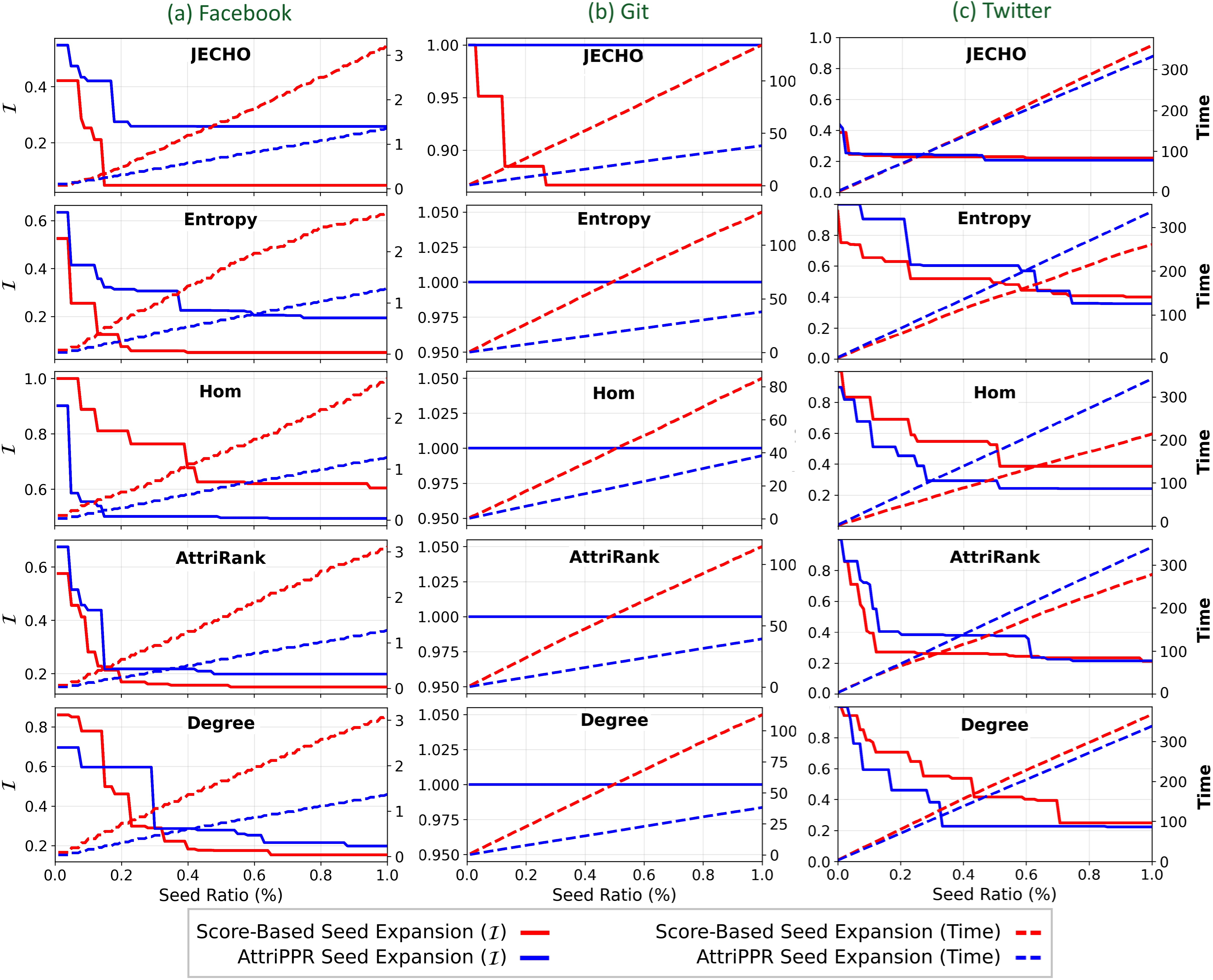}
 \caption{Comparison of SSE and AttriPPR expansion, showing structural isolation (solid) and runtime (dashed).}
 \Description{Plots comparing SSE and AttriPPR seed expansion on different datasets, showing structural isolation with solid lines and runtime with dashed lines under different seed scoring functions.}
 \label{fig:comparision-expansions}
\end{figure}

\begin{figure}[!htbp]
 \centering
 \includegraphics[width=0.95\linewidth]{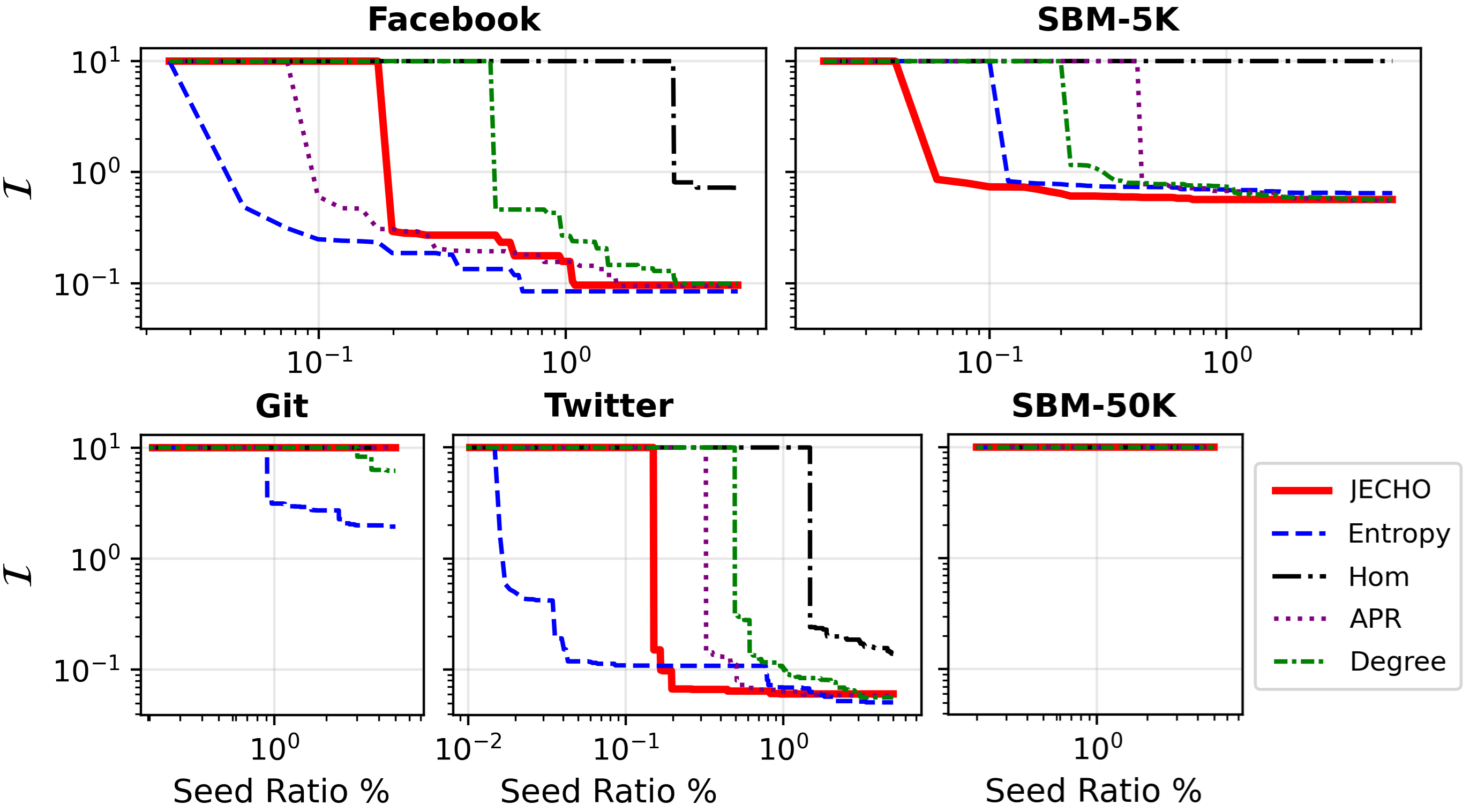}
 \caption{Comparison of seed scoring functions. Each plot shows $\mathcal{I}$ under different seed ratios using SSE.}
 \Description{Line plot showing structural isolation versus seed ratio on different datasets for different seed scoring functions, including JHO, entropy, homophily, degree, and AttriRank.}
\label{fig:score_seed_comparison_all2}
\end{figure}

\section*{Generative AI Usage Disclosure}
We used ChatGPT solely for language polishing and clarity improvement of text written by the authors. No scientific content, analysis, or results were generated using AI tools.

\bibliographystyle{ACM-Reference-Format}
\bibliography{sample-base}

\end{document}